\documentclass[journal]{IEEEtran}
\usepackage{amsmath,amsfonts}
\usepackage{algorithmic}
\usepackage{algorithm}
\usepackage{array}
\usepackage[caption=false,font=normalsize,labelfont=sf,textfont=sf]{subfig}
\usepackage{textcomp}
\usepackage{stfloats}
\usepackage{url}
\usepackage{verbatim}
\usepackage{graphicx}
\usepackage{cite}
\usepackage{bm}
\usepackage{amssymb}
\usepackage{mdwmath}
\usepackage{mdwtab}
\usepackage{eqparbox}
\usepackage{cuted}
\usepackage[hidelinks]{hyperref}

\newtheorem{proposition}{\bfseries Proposition}

\newtheorem{remark}{\bfseries Remark}

\begin{document}

\title{Rotatable Antenna Enabled Multi-Satellite Communications: Joint Satellite Selection and Boresight Trajectory Optimization}

\author{Xingxiang Peng, Qingqing Wu, Haiying Hu, Wen Chen, and Xin Lin
\thanks{
    X. Peng, Q. Wu, and W. Chen are with the School of Integrated Circuits, Shanghai Jiao Tong University, Shanghai 200240, China (e-mail:
    \href{mailto:peng_xingxiang@sjtu.edu.cn}{\nolinkurl{peng_xingxiang@sjtu.edu.cn}};
    \href{mailto:qingqingwu@sjtu.edu.cn}{\nolinkurl{qingqingwu@sjtu.edu.cn}};
    \href{mailto:wenchen@sjtu.edu.cn}{\nolinkurl{wenchen@sjtu.edu.cn}}).}
\thanks{
    H. Hu is with Innovation Academy for Microsatellites of Chinese Academy of Science, Shanghai 201203, China (e-mail: 
    \href{mailto:huhy@microsate.com}{\nolinkurl{huhy@microsate.com}}). }
\thanks{
    X. Lin is with Shanghai Institute of Satellite Engineering, Shanghai 201109, China (e-mail: 
    \href{mailto:scar07@sina.com}{\nolinkurl{scar07@sina.com}}). 
    }
}

\IEEEaftertitletext{\vspace{-2\baselineskip}}

\maketitle

\begin{abstract}
This paper considers a satellite-to-ground communication system in which a ground station (GS) equipped with independently rotatable antenna (RA) elements jointly decodes independent streams from multiple low-Earth-orbit (LEO) satellites over a shared time--frequency resource. Specifically, we formulate a two-timescale throughput maximization problem under exogenous cochannel interference, capturing serving-set composition, RA-enabled channel shaping, time-varying satellite geometry, and mechanically constrained inter-epoch reconfiguration. We first characterize the joint effects of interference-whitened channel strength and spatial separability on multi-satellite reception, motivating the joint design of satellite selection and RA control. With the RA trajectory fixed, we establish the monotone submodularity of the epoch-level selection objective and construct an incumbent-tight modular lower-bound surrogate, leading to an efficient discrete Minorization-Maximization (MM) selection algorithm. For fixed serving sets, we develop slew-feasible RA updates based on Riemannian gradients and organize them into a two-color parallel update scheme. The two blocks are integrated into a monotone alternating algorithm with guaranteed objective convergence. Simulations demonstrate consistent gains over benchmark schemes and reveal an optimal balance between channel strength and spatial separability. The results further show that satellite selection is particularly important in underloaded and actuator-limited regimes, whereas RA shaping becomes more influential near full spatial loading.
\end{abstract}


\section{Introduction}\label{sec:introduction}

\IEEEPARstart{U}{biquitous} connectivity, service continuity, and global access envisioned for sixth-generation (6G) networks are accelerating the integration of terrestrial infrastructure with non-terrestrial networks (NTNs) \cite{9861699}. Among the available NTN platforms, low-Earth-orbit (LEO) constellations are particularly promising because their shorter propagation distances can reduce latency relative to geostationary systems, while dense satellite deployments support wide-area coverage and provide rich spatial and network-level degrees of freedom \cite{9852737}. These advantages, however, are accompanied by highly dynamic network topology and intensive spectrum reuse, which complicate resource coordination and intensify inter- and intra-system interference. Fully realizing the potential of large-scale LEO constellations therefore requires system-level design and coordination strategies that can exploit their multidimensional diversity while adapting to orbital evolution.

Accordingly, a broad range of advanced signal-processing, resource-management, and network-coordination techniques has been developed for highly dynamic LEO environments. For instance, beam hopping, illumination-pattern design, and spatiotemporal scheduling enable dynamic radio-resource adaptation to time-varying and nonuniform traffic demands \cite{9968247,10694728,10486925}, while hybrid beamforming, coordinated transmission, and cell-free NTN architectures exploit multi-satellite cooperation \cite{10380500,10679987,10787138}. At the infrastructure level, gateway planning and feeder-link switching leverage geographical and temporal diversity \cite{10273398,11316666}, whereas feeder-link MIMO and multigateway processing improve spatial multiplexing and interference management \cite{9781427,9684855}. Interference mitigation, coexistence analysis, and interference-aware satellite selection further address spectrum sharing and connectivity in dense LEO deployments \cite{10596023,10797644,11356005}. Despite their diverse objectives and operating layers, these approaches primarily optimize network control and signal-processing strategies over channels determined by the propagation environment. Their achievable performance may therefore remain limited by unfavorable channel characteristics, including strong spatial correlation, limited effective rank, and directional mismatch.

To move beyond optimizing communication strategies over externally determined channels, recent studies have explored physical-domain reconfigurability to actively reshape the effective wireless channel. Intelligent reflecting surfaces (IRSs) provide a representative environment-side solution by introducing programmable reflected paths and have been applied to multi-satellite downlink beamforming, joint satellite scheduling with active/passive beamforming, and cooperative multi-satellite uplink reception \cite{10494510,11251206,11454451}. Movable antennas (MAs), by contrast, exploit transceiver-side spatial reconfigurability by adjusting antenna positions and the resulting array geometry \cite{10318061}. In satellite communications, the element positions of a satellite-mounted MA array have been jointly optimized with time-varying beamforming to adapt its coverage pattern and suppress signal leakage toward unintended regions \cite{10806489}. More recently, an MA-equipped LEO ground station has been investigated by jointly designing antenna positions and time-varying beamforming to improve average communication performance while mitigating interference in dense satellite deployments \cite{11303331}. These MA studies demonstrate the potential of transceiver-side geometric reconfiguration in both spaceborne and ground-based satellite architectures.

Complementary to antenna-position reconfiguration, making element orientations reconfigurable introduces an additional transceiver-side physical degree of freedom by exploiting practical antenna directivity. Rotatable-antenna (RA) technology harnesses this freedom by controlling orientation-dependent element gains before baseband processing, thereby reshaping the effective channels and redistributing array sensitivity across desired and interfering directions \cite{11427014,zheng2026rotatable,peng2026,11568555}. This capability is especially promising for LEO communications, where far-field, line-of-sight-dominated links are strongly governed by angular geometry, while orbital motion makes desired and interfering directions evolve continuously yet predictably. RA can thus turn directional dynamics from a source of channel mismatch into a resource for proactive channel shaping. Recent studies have explored RA-assisted designs for multi-user uplink, integrated sensing and communication, interference mitigation, and cell-free transmission \cite{11520842,11520277,11489290,11635960}. In LEO satellite communications, a recent study \cite{Ma2026RASatellite} investigates a point-to-point link with RA arrays deployed at both the satellite and ground terminals, demonstrating the benefits of per-element boresight alignment and orbital-motion-assisted tracking for a single rapidly varying satellite direction. This single-link formulation, however, does not address serving-set composition, concurrent reception of independent satellite streams, or joint channel shaping across multiple desired and interfering directions at a common ground station.

Against this background, we consider an RA-enhanced satellite-to-ground communication system in which a ground station (GS) selects a serving set of LEO satellites and jointly decodes their independent streams over a shared time--frequency resource. Satellite selection determines the composition of the received multi-satellite channel, whereas RA steering reshapes its effective link gains, spatial separability, and susceptibility to external interference. Their mutual dependence motivates the joint optimization of satellite selection and RA control. In particular, directionally clustered satellites may share favorable RA gains but produce highly correlated channel responses, whereas a more spatially diverse serving set may improve spatial separability at the cost of reduced directional-gain sharing. Thus, multi-satellite reception is governed jointly by effective channel strength and spatial separability, with external cochannel interference further altering this balance in the interference-whitened domain. The resulting design naturally exhibits a two-timescale structure, with the serving set and RA boresights held fixed within each control epoch, while satellite geometry and channel coefficients vary at the slot level and the digital receive processing is updated accordingly. Moreover, finite RA actuation speed couples the boresight configurations across adjacent epochs, introducing mechanical trajectory constraints into the joint design. The main contributions of this paper are summarized as follows.

\begin{itemize}
    \item We develop a two-timescale framework for RA-enhanced satellite-to-ground communications under exogenous cochannel interference and formulate the joint optimization of epoch-held serving sets and RA boresight trajectories to maximize throughput. The formulation captures the coupling among serving-set composition, RA-enabled channel shaping, time-varying satellite geometry, and mechanically constrained inter-epoch reconfiguration.

    \item We characterize the interference-whitened strength--separability structure of RA-enhanced multi-satellite reception, showing that each candidate satellite's marginal capacity contribution depends jointly on its effective channel strength and its spatial complementarity with the selected channels. We further establish that, with the RA trajectory fixed, the epoch-level satellite-selection objective is monotone and submodular, revealing a diminishing-return structure that enables efficient capacity-aware selection.

    \item We develop an efficient and highly parallelizable alternating framework for joint satellite selection and boresight trajectory design. For satellite selection, we construct an incumbent-tight modular lower-bound surrogate and develop a discrete Minorization-Maximization (MM) algorithm, with all epoch-level problems solved in parallel. For RA-trajectory optimization, we use Riemannian gradients to construct slew-feasible local updates and exploit the chain-structured inter-epoch coupling through a two-color parallel update scheme. The resulting alternating updates generate a nondecreasing throughput and guarantee convergence of the objective value.

    \item Extensive simulations validate both the strength--separability interpretation and the effectiveness of the proposed joint design. The results reveal an optimal balance between channel strength and spatial separability, rather than either strongest-link clustering or maximum angular separation. They also show that satellite selection is particularly important in underloaded and actuator-limited regimes, whereas RA shaping becomes more influential near full spatial loading. The complementarity between satellite selection and RA shaping becomes more pronounced under severe external interference.
\end{itemize}

The remainder of this paper is organized as follows. Section~\ref{sec:system_model} presents the two-timescale system model and formulates the joint satellite-selection and boresight-trajectory optimization problem. Section~\ref{sec:algorithm} develops the strength--separability characterization and the proposed joint optimization framework, together with its convergence and complexity analyses. Section~\ref{sec:numerical_results} presents the numerical results and discusses the resulting insights, while Section~\ref{sec:conclusion} concludes the paper.

% \emph{Notations:} Scalars, vectors, and matrices are denoted by italic letters, boldface lowercase letters, and boldface uppercase letters, respectively. The sets of real and complex numbers are denoted by $\mathbb{R}$ and $\mathbb{C}$, respectively. The operators $(\cdot)^T$, $(\cdot)^*$, $(\cdot)^H$, and $(\cdot)^\dagger$ denote transpose, complex conjugate, conjugate transpose, and pseudoinverse, respectively. Moreover, $\|\cdot\|_2$, $\operatorname{tr}(\cdot)$, $\det(\cdot)$, and $\Re\{\cdot\}$ denote the Euclidean norm, trace, determinant, and real part, respectively. The imaginary unit is denoted by $\mathrm{j}\triangleq\sqrt{-1}$. The matrix $\bm{I}_d$ denotes the $d\times d$ identity matrix, with the dimension omitted when clear from context, while $\operatorname{diag}(\cdot)$ constructs a diagonal matrix. For a vector $\bm{x}$, $[\bm{x}]_m$ denotes its $m$th entry, while for a set $\mathcal{S}$, $|\mathcal{S}|$ denotes its cardinality. In addition, $[x]_+\triangleq\max\{x,0\}$, $\mathbf{1}\{\cdot\}$ denotes the indicator function, and $\mathcal{CN}(\bm{0},\bm{\Sigma})$ denotes a circularly symmetric complex Gaussian distribution with covariance matrix $\bm{\Sigma}$.

\section{System Model and Problem Formulation}
\label{sec:system_model}

As illustrated in Fig.~\ref{fig:sysModel}, we consider a satellite-to-ground communication system in which an RA-equipped GS concurrently receives independent streams from multiple LEO satellites over a shared time--frequency resource. The system comprises a serving constellation and an external constellation. Specifically, at each control epoch, the GS selects a subset of satellites from the serving constellation for concurrent transmission, while the remaining satellites stay inactive on the considered resource. By contrast, the external constellation operates independently of the serving constellation, and its satellites may transmit concurrently over the same resource, thereby contributing exogenous cochannel interference.

\begin{figure}[t]
    \centering
    \includegraphics[width=0.45\textwidth]{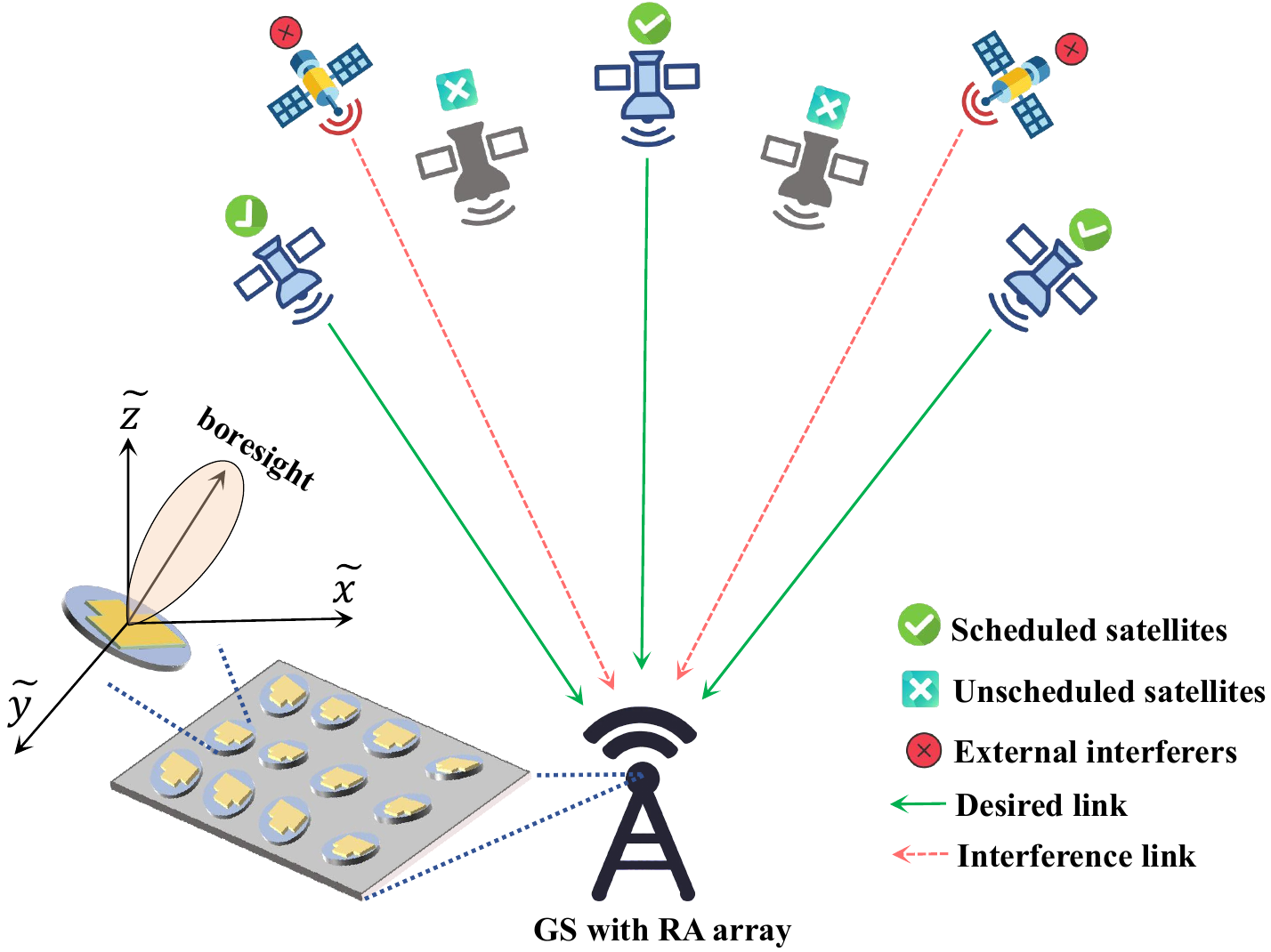}
    \caption{Illustration of the RA-enhanced satellite-to-ground communication system with satellite selection and external cochannel interference.}
    \label{fig:sysModel}
\end{figure}

\subsection{Two-Timescale Operation Protocol}
\label{subsec:two_timescale}

To account for the slower mechanical reconfiguration of the RAs relative to digital receive processing, we adopt the two-timescale protocol illustrated in Fig.~\ref{fig:two_timescale}. Satellite selection and RA reconfiguration are performed at the epoch level, whereas the channel state is evaluated and the digital receiver is updated at the slot level. The observation period contains $N$ data-bearing slots of duration $\Delta t$, indexed by $\mathcal{N}\triangleq\{1,\ldots,N\}$, which are partitioned into $L$ control epochs. Each epoch contains $Q\triangleq N/L$ consecutive slots, with
\begin{align}
    \mathcal{N}_\ell &\triangleq \{(\ell-1)Q+1,\ldots,\ell Q\}, \quad \ell=1,\ldots,L.
\end{align}
Each epoch $\ell\geq2$ is preceded by a guard interval of duration $T_{\rm g}$, during which payload reception is suspended while the next serving set is activated and the RA elements slew to their target boresights. For the first epoch, however, the configuration is established before $t=0$ and therefore incurs no guard overhead within the observation period. Hence, the total observation duration $T_{\rm ob}$ and the midpoint time $t_n$ of data-bearing slot $n$ are given by
\begin{subequations}
\begin{align}
    T_{\rm ob} &=N\Delta t+(L-1)T_{\rm g}, \label{eq:slot_time1}\\
    t_n &=\left(n-\frac{1}{2}\right)\Delta t +\left(\left\lceil\frac{n}{Q}\right\rceil-1\right)T_{\rm g},
    \label{eq:slot_time}
\end{align}
\end{subequations}
where the second term in $t_n$ accounts for the guard intervals preceding slot~$n$. Under this protocol, the serving set and RA boresights are held fixed throughout each control epoch. Within an epoch, the satellite geometry and channel coefficients are evaluated at each slot midpoint, and the minimum mean-square error (MMSE) receiver with successive interference cancellation (SIC) is performed using the instantaneous channel state information (CSI) \cite{Tse2005Fundamentals}.

\begin{figure}[t]
    \centering
    \includegraphics[width=0.45\textwidth]{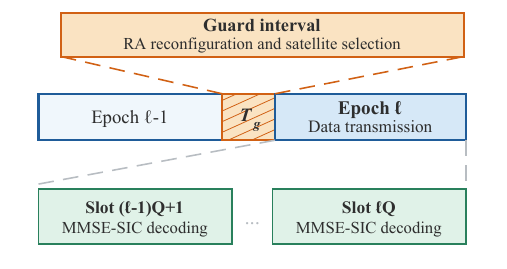}
    \caption{Illustration of the two-timescale operation: the serving satellite set and RA boresights are updated during each guard interval and remain fixed throughout the subsequent epoch, whereas the digital receiver is updated in every data-bearing slot.}
    \label{fig:two_timescale}
\end{figure}

\subsection{Dynamic Satellite Geometry and Candidate Satellite Sets}
\label{subsec:satellite_geometry}

We consider a representative large-scale LEO deployment based on a Walker--Delta constellation, specified by $\iota:S/J/F_{\rm W}$ \cite{Walker1984SatelliteConstellations}. Here, $\iota$ is the common orbital inclination, $S$ is the total number of satellites, $J$ is the number of circular orbital planes, and $F_{\rm W}\in\{0,\ldots,J-1\}$ is the inter-plane phasing parameter. Each plane contains $K\triangleq S/J$ uniformly spaced satellites. We index the orbital planes by $j\in\mathcal{J}\triangleq\{0,\ldots,J-1\}$ and the satellites within each plane by $k\in\mathcal{K}\triangleq\{0,\ldots,K-1\}$, and denote the corresponding satellite by $(j,k)$. Assuming that all satellites orbit at a common altitude $H$, the orbital radius and mean motion are given by $R_{\rm o}=R_{\rm E}+H$ and $\omega_{\rm o}=\sqrt{\mu/R_{\rm o}^{3}}$, respectively, where $R_{\rm E}$ and $\mu$ denote the Earth radius and gravitational parameter, respectively. The right ascension of the ascending node (RAAN) of plane $j$ and the argument of latitude of satellite $(j,k)$ are given by
\begin{subequations}
\begin{align}
    \Omega_j &=\Omega_0+\frac{2\pi j}{J}, \label{eq:raan}\\
    u_{j,k}(t) &=u_0+\omega_{\rm o}t+\frac{2\pi k}{K} +\frac{2\pi F_{\rm W}j}{JK},
    \label{eq:argument_latitude}
\end{align}
\end{subequations}
where $\Omega_0$ denotes the reference RAAN and $u_0$ is the initial orbital phase. Let $\bm{R}_x(\cdot)$ and $\bm{R}_z(\cdot)$ denote the standard right-handed rotation matrices about the $x$- and $z$-axes, respectively. With these definitions, the Earth-centered inertial (ECI) and Earth-centered Earth-fixed (ECEF) positions of satellite $(j,k)$ are given by
\begin{subequations}
\begin{align}
    \bm{r}_{j,k}^{\rm I}(t) &=\bm{R}_z(\Omega_j)\bm{R}_x(\iota)R_{\rm o} [\cos u_{j,k}(t),\sin u_{j,k}(t),0]^T, \label{eq:satellite_eci}\\
    \bm{r}_{j,k}^{\rm E}(t) &=\bm{R}_z(-\theta_{\rm G,0}-\omega_{\rm E}t) \bm{r}_{j,k}^{\rm I}(t),
    \label{eq:satellite_ecef}
\end{align}
\end{subequations}
where $\theta_{\rm G,0}$ and $\omega_{\rm E}$ are the initial Greenwich sidereal angle and Earth rotation rate, respectively.

Let $(\varphi_{\rm g},\lambda_{\rm g},h_{\rm g})$ denote the geocentric latitude, longitude, and altitude of the GS. Under a spherical-Earth model, its ECEF position is given by
\begin{equation}
    \bm{r}_{\rm g}^{\rm E} =(R_{\rm E}+h_{\rm g})
    \begin{bmatrix}
    \cos\varphi_{\rm g}\cos\lambda_{\rm g}\\
    \cos\varphi_{\rm g}\sin\lambda_{\rm g}\\
    \sin\varphi_{\rm g}
    \end{bmatrix}.
    \label{eq:ground_station_ecef}
\end{equation}
The ECEF-to-east--north--up (ENU) transformation at the GS is given by
\begin{equation}
    \bm{T}_{\rm g} =
    \begin{bmatrix}
    -\sin\lambda_{\rm g} & \cos\lambda_{\rm g} & 0\\
    -\sin\varphi_{\rm g}\cos\lambda_{\rm g} &-\sin\varphi_{\rm g}\sin\lambda_{\rm g} &\cos\varphi_{\rm g}\\
    \cos\varphi_{\rm g}\cos\lambda_{\rm g} &\cos\varphi_{\rm g}\sin\lambda_{\rm g} &\sin\varphi_{\rm g}
    \end{bmatrix}.
    \label{eq:ecef_to_enu}
\end{equation}
Hence, the GS-centered relative position, slant range, unit direction, and elevation angle of satellite $(j,k)$ are respectively obtained as
\begin{subequations}
\begin{align}
    \bm{\rho}_{j,k}(t) &= \bm{T}_{\rm g} \left(\bm{r}_{j,k}^{\rm E}(t)-\bm{r}_{\rm g}^{\rm E}\right), \label{eq:relative_enu}\\
    d_{j,k}(t) &=\|\bm{\rho}_{j,k}(t)\|_2, \label{eq:slant_range}\\
    \bm{d}_{j,k}(t) &=\frac{\bm{\rho}_{j,k}(t)}{d_{j,k}(t)}, \label{eq:unit_direction}\\
    \epsilon_{j,k}(t) &= \operatorname{atan2}\!\left( [\bm{\rho}_{j,k}(t)]_3, \sqrt{[\bm{\rho}_{j,k}(t)]_1^2+[\bm{\rho}_{j,k}(t)]_2^2} \right).
    \label{eq:elevation}
\end{align}
\end{subequations}
For slot-level modeling, the corresponding slant range, unit direction, and elevation angle evaluated at the slot midpoint $t_n$ are denoted by $d_{j,k}[n]$, $\bm{d}_{j,k}[n]$, and $\epsilon_{j,k}[n]$, respectively. A satellite is visible in slot $n$ if
\begin{equation}
    v_{j,k}[n] \triangleq \mathbf{1}\{\epsilon_{j,k}[n]\geq\epsilon_{\min}\} =1,
    \label{eq:visibility_indicator}
\end{equation}
where $\epsilon_{\min}$ is the minimum elevation angle. Since the serving set remains fixed within each control epoch, a satellite is retained as a candidate only if it satisfies the elevation requirement throughout that epoch. Accordingly, the candidate set for epoch $\ell$ is constructed as
\begin{equation}
    \mathcal{A}_\ell \triangleq \left\{(j,k):v_{j,k}[n]=1,\quad \forall n\in\mathcal{N}_\ell\right\}.
    \label{eq:epoch_candidate_set}
\end{equation}

\subsection{RA Reconfiguration and Channel Model}
\label{subsec:ra_channel}

\subsubsection{RA Configuration and Directional Gain}
The GS employs an $M_x\times M_y$ uniform planar array (UPA) comprising $M=M_xM_y$ RA elements. We align the UPA coordinate system with the GS-centered ENU frame, such that the array lies in the local horizontal plane with its $x$-, $y$-, and $z$-axes pointing east, north, and up, respectively. For $m_x\in\{0,\ldots,M_x-1\}$ and $m_y\in\{0,\ldots,M_y-1\}$, define the element index as $m=m_x+m_yM_x+1$. The position of element $m$ relative to the UPA center is then given in the ENU frame by
\begin{equation}
    \bm{p}_m =
    \begin{bmatrix}
    \left(m_x-\frac{M_x-1}{2}\right)d_x,\, \left(m_y-\frac{M_y-1}{2}\right)d_y,\, 0
    \end{bmatrix}^T,
    \label{eq:element_position}
\end{equation}
where $d_x$ and $d_y$ are the inter-element spacings. Each RA element $m$ is configured with an epoch-dependent boresight $\bm{f}_m[\ell]$, expressed in the GS-centered ENU frame. Let $\bm{e}_z=[0,0,1]^T$ denote the nominal array normal. The feasible steering region is represented by a spherical cap \cite{11489290}
\begin{equation}
    \mathcal{F}_{\rm RA} \triangleq \left\{ \bm{f}\in\mathbb{R}^3: \|\bm{f}\|_2=1,\, \bm{f}^T\bm{e}_z\geq\cos\theta_{\max} \right\},
    \label{eq:ra_feasible_set}
\end{equation}
where $\theta_{\max}\in[0,\pi/2)$ is the maximum steering angle from the nominal array normal $\bm{e}_z$. In addition to this steering-range constraint, the finite RA actuation speed limits the boresight variation across adjacent epochs. Let $\omega_{\max}$ denote the maximum angular speed. During a guard interval of duration $T_{\rm g}$, the maximum allowable angular displacement of each boresight is
\begin{equation}
    \Omega_{\rm g} \triangleq \omega_{\max}T_{\rm g}.
\end{equation}
Accordingly, consecutive boresights are constrained by
\begin{equation}
    \bm{f}_m^T[\ell-1]\bm{f}_m[\ell] \geq\cos\Omega_{\rm g}, \quad \forall m,\ \ell\geq2.
    \label{eq:slew_constraint}
\end{equation}
Note that the epoch-$1$ boresights are pre-positioned before the observation interval and are therefore not subject to an initial slew constraint. Given the boresight configuration, we adopt the front-side cosine-power pattern in \cite{Ant2016} to model the direction-dependent receive gain. For RA element $m$ and satellite $(j,k)$ in slot $n\in\mathcal{N}_\ell$, the receive antenna gain is
\begin{equation}
    G_{m,j,k}^{\rm rx}[n] = \kappa_{\max} \left[ \bm{f}_m^T[\ell]\bm{d}_{j,k}[n] \right]_+^{2p},
    \label{eq:ra_gain}
\end{equation}
where $p \geq 0$ is the directivity exponent, and $\kappa_{\max}=2(2p+1)$ normalizes the average gain over the sphere to unity. Increasing $p$ narrows the main lobe and increases the peak gain.

\subsubsection{Satellite-to-GS Channel Model}
We next combine the receive-side RA response with transmit-side directivity, large-scale propagation, and array-induced phase differences to obtain the satellite-to-GS channel. When scheduled for transmission, satellite $(j,k)$ is assumed to steer its transmit beam toward the GS, with $G_{j,k}^{\rm tx}[n]$ denoting the corresponding transmit antenna gain. In addition, the large-scale propagation gain is modeled as \cite{Ant2016}
\begin{equation}
    \beta_{j,k}[n] = \left(\frac{\lambda}{4\pi d_{j,k}[n]}\right)^2 \chi_{j,k}[n],
    \label{eq:path_gain}
\end{equation}
where $\lambda$ is the carrier wavelength, and $\chi_{j,k}[n]\in(0,1]$ collects rain attenuation, gaseous absorption, and other slowly varying losses. Because LEO satellites lie in the far field of the GS array, the satellite-to-element propagation distance can be approximated by
\begin{align}
    r_{m,j,k}[n]\simeq d_{j,k}[n]-\bm{p}_m^T\bm{d}_{j,k}[n].
\end{align}
With these components, we define the UPA phase-response vector, RA-dependent amplitude matrix, and common scalar channel coefficient as
\begin{subequations}
\begin{align}
    \bm{a}_{j,k}[n] &\triangleq \left[ e^{\mathrm{j}\frac{2\pi}{\lambda} \bm{p}_m^T\bm{d}_{j,k}[n]} \right]_{m=1}^{M}, \label{eq:upa_response}\\
    \bm{\Gamma}_{j,k}[\ell,n] &\triangleq \operatorname{diag}\left( \left\{ \sqrt{\kappa_{\max}} [\bm{f}_m^T[\ell]\bm{d}_{j,k}[n]]_+^p \right\}_{m=1}^{M} \right), \label{eq:ra_amplitude_matrix}\\
    c_{j,k}[n] &\triangleq \sqrt{\beta_{j,k}[n]G_{j,k}^{\rm tx}[n]} e^{-\mathrm{j}\frac{2\pi}{\lambda}d_{j,k}[n]}.
    \label{eq:common_channel_scalar}
\end{align}
\end{subequations}
Accordingly, the satellite-to-GS channel vector is given by
\begin{equation}
    \bm{h}_{j,k}[n] = c_{j,k}[n] \bm{\Gamma}_{j,k}[\ell,n] \bm{a}_{j,k}[n], \quad n\in\mathcal{N}_\ell.
    \label{eq:channel_decomposition}
\end{equation}
The time-varying satellite geometry governs the propagation gain and array phase response, while the RA boresights provide additional element-wise control over the direction-dependent gains. For notational simplicity, we henceforth use a single index $s$ to denote each satellite $(j,k)$.

\subsection{Multi-Satellite Reception and Sum Capacity}
\label{subsec:multi_satellite_reception}

Let $\mathcal{S}_\ell\subseteq\mathcal{A}_\ell$ denote the serving set selected in epoch $\ell$, with $|\mathcal{S}_\ell|\leq K_{\max}$, where $K_{\max}$ is the maximum number of simultaneously scheduled streams. The GS employs a fully digital receiver with one RF chain per RA element. The external constellation operates independently of the serving constellation and is not controlled by the GS. Let $\mathcal{I}_\ell$ denote the external cochannel interferer set in epoch $\ell$. For each $q\in\mathcal{I}_\ell$, the effective interference channel $\bm{h}_q[n]$ follows the same receive-side channel construction as in \eqref{eq:channel_decomposition}, with its own satellite geometry and link parameters. In particular, it depends on the RA boresights through the corresponding direction-dependent receive gains. The arrival directions and large-scale channel coefficients of the external interferers are assumed available at the GS through interference sensing and short-term prediction.

For each scheduled satellite $s\in\mathcal{S}_\ell$, let $P_s[n]$ denote its transmit power. For each external interferer $q\in\mathcal{I}_\ell$, let $P_q[n]$ denote its effective in-band transmit power. Let $x_r[n]\sim\mathcal{CN}(0,1)$ denote the independent symbol transmitted by source $r\in\mathcal{S}_\ell\cup\mathcal{I}_\ell$. The received signal in slot $n\in\mathcal{N}_\ell$ is then given by
\begin{align}
    \bm{y}[n] =&\sum_{s\in\mathcal{S}_\ell}\sqrt{P_s[n]}\bm{h}_s[n]x_s[n]\notag\\
    &+ \sum_{q\in\mathcal{I}_\ell} \sqrt{P_q[n]}\bm{h}_q[n]x_q[n] +\bm{n}[n],
    \label{eq:received_signal}
\end{align}
where $\bm{n}[n]\sim\mathcal{CN}(\bm{0},\sigma^2\bm{I}_M)$ is the receiver noise. The interference-plus-noise covariance matrix and total received covariance matrix are respectively given by
\begin{subequations}
\begin{align}
    \bm{R}_0[n] &\triangleq \sum_{q\in\mathcal{I}_\ell} P_q[n]\bm{h}_q[n]\bm{h}_q^H[n] +\sigma^2\bm{I}_M, \label{eq:base_covariance}\\
    \bm{R}_{\mathcal{S}_\ell}[n] &\triangleq \bm{R}_0[n] +\sum_{s\in\mathcal{S}_\ell} P_s[n]\bm{h}_s[n]\bm{h}_s^H[n].
    \label{eq:total_covariance}
\end{align}
\end{subequations}
With capacity-achieving MMSE-SIC reception, the slot sum capacity is expressed as \cite{Tse2005Fundamentals}
\begin{equation}
    C_{\mathcal{S}_\ell}[n] = B\left( \log_2\det\bm{R}_{\mathcal{S}_\ell}[n] -\log_2\det\bm{R}_0[n] \right),
    \label{eq:mac_sum_capacity}
\end{equation}
where $B$ is the resource bandwidth. Note that the SIC order affects the per-satellite rate allocation but not the sum capacity. Consistent with the two-timescale protocol, the MMSE-SIC receiver is updated at the slot level using instantaneous CSI, whereas satellite selection and RA control are performed at the epoch level based on predicted satellite geometry and large-scale channel and interference information.

\begin{remark}
\label{rem:kmax_choice}
With MMSE-SIC reception, the number of simultaneously scheduled satellites may exceed the number of GS receive antennas $M$, and the sum capacity can still increase as additional satellites are scheduled. However, streams beyond $M$ do not provide additional spatial degrees of freedom, and thus at high signal-to-noise ratio (SNR), the sum capacity grows only logarithmically with the number of scheduled satellites \cite{Tse2005Fundamentals}. Moreover, scheduling more satellites incurs increased channel-acquisition and decoding overhead.
\end{remark}

\subsection{Problem Formulation}
\label{subsec:problem}
Accounting for the guard intervals, the effective throughput over the observation period is defined as
\begin{equation}
    \mathcal{T}(\{\mathcal{S}_\ell\},\bm{F}_{\rm RA}) \triangleq \frac{\Delta t}{T_{\rm ob}} \sum_{\ell=1}^{L} \sum_{n\in\mathcal{N}_\ell} C_{\mathcal{S}_\ell}[n],
    \label{eq:effective_throughput}
\end{equation}
where $\bm{F}_{\rm RA}[\ell]\triangleq[\bm{f}_1[\ell],\ldots,\bm{f}_M[\ell]]$ denotes the RA configuration in epoch $\ell$, and $\bm{F}_{\rm RA}\triangleq\{\bm{F}_{\rm RA}[\ell]\}_{\ell=1}^{L}$ denotes the complete RA trajectory. While the serving sets and RA boresights remain fixed within each epoch, the satellite geometry and effective channels vary across its constituent slots. Hence, \eqref{eq:effective_throughput} evaluates each epoch-level decision over the corresponding intra-epoch geometry evolution. We jointly optimize the epoch-level serving sets and RA trajectory as
\begin{subequations}
\label{prob:joint_design}
\begin{align}
    \max_{\{\mathcal{S}_\ell\},\,\bm{F}_{\rm RA}} \quad &\mathcal{T}(\{\mathcal{S}_\ell\},\bm{F}_{\rm RA}) \label{prob:joint_design_obj}\\
    \mathrm{s.t.}\qquad &\bm{f}_m[\ell]\in\mathcal{F}_{\rm RA}, \quad \forall m,\ell, \label{prob:orientation_constraints}\\
    &\bm{f}_m^T[\ell-1]\bm{f}_m[\ell] \geq\cos\Omega_{\rm g}, \quad \forall m,\ \ell \geq 2, \label{prob:slew_constraints}\\
    &\mathcal{S}_\ell\subseteq\mathcal{A}_\ell,\quad |\mathcal{S}_\ell|\leq K_{\max}, \quad \forall\ell.
    \label{prob:selection_constraints}
\end{align}
\end{subequations}
Constraint~\eqref{prob:orientation_constraints} restricts each RA boresight to its feasible steering region. Constraint~\eqref{prob:slew_constraints} limits the angular displacement between consecutive epochs according to the available reconfiguration time, while constraint~\eqref{prob:selection_constraints} requires each serving set to be drawn from the corresponding candidate set and limits the number of simultaneously scheduled streams. Problem~\eqref{prob:joint_design} is a mixed discrete--continuous nonconvex problem, where satellite selection determines which channels contribute to the sum capacity, while the RA boresights reshape the desired and interfering channel gains. The slew constraints also couple the RA configurations across adjacent epochs.

\begin{remark}
    Since the selected satellite streams are jointly decoded and each satellite is subject to an independent peak-power constraint, the sum capacity is monotonically nondecreasing in transmit power. Hence, optimizing the transmit powers would set all scheduled satellites to their maximum power levels. We therefore fix them at their maximum values and focus on satellite selection and RA control.
\end{remark}

\section{Efficient Joint Satellite Selection and RA-Trajectory Optimization}
\label{sec:algorithm}

In this section, we first reveal the strength–separability tradeoff inherent in multi-satellite reception, which motivates the joint design of satellite selection and RA control. We then develop a discrete MM method for satellite selection based on an incumbent-tight modular surrogate, and a slew-feasible RA-trajectory optimization method based on a two-color parallel update scheme. Finally, we integrate the two optimization blocks into a monotone alternating algorithm and analyze its convergence and computational complexity.

\subsection{Geometric Insights for Multi-Satellite Reception}
\label{subsec:capacity_geometry}

We first consider two simultaneously received satellites to reveal how effective channel strength and spatial separability jointly shape the slot sum capacity. For notational simplicity, the slot index is omitted throughout this subsection. Define the interference-whitened effective channels as
\begin{equation}
    \bm{g}_s \triangleq \sqrt{P_s}\bm{R}_0^{-1/2}\bm{h}_s, \quad s\in\{1,2\}.
    \label{eq:whitened_channels}
\end{equation}
Since $\bm{R}_0\succ\bm{0}$, this whitening transformation is invertible and yields an equivalent representation in which the interference-plus-noise has identity covariance. Consequently, the Euclidean norms and mutual alignment of the whitened channels provide interference-aware measures of their effective strengths and spatial separability. Under this representation, the slot sum capacity in \eqref{eq:mac_sum_capacity} can be equivalently expressed as
\begin{equation}
    C_{\{1,2\}} = B\log_2\det\left( \bm{I}_M+\bm{g}_1\bm{g}_1^H+\bm{g}_2\bm{g}_2^H \right).
    \label{eq:mac_sum_capacity2}
\end{equation}
The interference-whitened channel strength of satellite $s$ is quantified by
\begin{equation}
    \alpha_s \triangleq \|\bm{g}_s\|_2^2, \quad s\in\{1,2\},
\end{equation}
while the spatial alignment between the two whitened channels is characterized by the coherence coefficient
\begin{equation}
    \bar\rho \triangleq \frac{|\bm{g}_1^H\bm{g}_2|} {\|\bm{g}_1\|_2\|\bm{g}_2\|_2} \in[0,1].
    \label{eq:whitened_coherence}
\end{equation}

\begin{proposition}[Capacity Decomposition]
\label{prop:two_satellite_capacity}
The slot sum capacity in \eqref{eq:mac_sum_capacity2} can be equivalently expressed as
\begin{equation}
    C_{\{1,2\}} = B\log_2\!\left( 1+\alpha_1+\alpha_2+ \alpha_1\alpha_2(1-\bar\rho^2) \right).
    \label{eq:two_satellite_capacity}
\end{equation}
\end{proposition}

\begin{IEEEproof}
Let $\bm{G}=[\bm{g}_1,\bm{g}_2]$. By Sylvester's determinant identity, $\det(\bm{I}_M+\bm{G}\bm{G}^H)=\det(\bm{I}_2+\bm{G}^H\bm{G})$. Since $\bm{I}_2+\bm{G}^H\bm{G}$ has diagonal entries $1+\alpha_1$ and $1+\alpha_2$ and off-diagonal entries $\bm{g}_1^H\bm{g}_2$ and $\bm{g}_2^H\bm{g}_1$, its determinant is $(1+\alpha_1)(1+\alpha_2)-|\bm{g}_1^H\bm{g}_2|^2$. Substituting $|\bm{g}_1^H\bm{g}_2|^2=\alpha_1\alpha_2\bar\rho^2$ into the determinant expression and using \eqref{eq:mac_sum_capacity2} yields \eqref{eq:two_satellite_capacity}.
\end{IEEEproof}

Proposition~\ref{prop:two_satellite_capacity} makes explicit how channel strength and spatial separability jointly determine the sum capacity. The terms $\alpha_1$ and $\alpha_2$ capture the individual interference-whitened channel strengths, whereas $\alpha_1\alpha_2(1-\bar\rho^2)$ captures the additional gain arising from their spatial complementarity. For fixed $\alpha_1$ and $\alpha_2$, this term decreases monotonically with the channel coherence $\bar\rho$. In the fully aligned case $\bar\rho=1$, it vanishes and the two channels contribute only through their aggregate strength, yielding $C_{\{1,2\}}=B\log_2(1+\alpha_1+\alpha_2)$. By contrast, when $\bar\rho=0$, the two whitened channels are orthogonal and the sum capacity becomes $C_{\{1,2\}}=B\log_2(1+\alpha_1)+B\log_2(1+\alpha_2)$, corresponding to two fully separable spatial modes. Hence, increasing channel strength is beneficial, but its capacity can be substantially reduced when the received channels are highly aligned. The relative importance of strength and separability further depends on the operating regime. At low SNR,
\begin{equation}
    C_{\{1,2\}} = \frac{B}{\ln 2}(\alpha_1+\alpha_2) +o(\alpha_1+\alpha_2),
    \label{eq:low_snr_insight}
\end{equation}
so the first-order capacity contribution is governed primarily by the individual channel strengths, while spatial separability has only a higher-order effect. At high SNR and for any $\bar\rho<1$,
\begin{equation}
    C_{\{1,2\}} \simeq B\log_2\!\left( \alpha_1\alpha_2(1-\bar\rho^2) \right),
    \label{eq:high_snr_insight}
\end{equation}
so spatial separability becomes inseparable from channel strength in determining the dominant capacity term. Therefore, selecting individually strong satellites does not necessarily yield a large multi-satellite sum capacity if their interference-whitened channels are highly aligned.

These observations motivate the joint optimization of satellite selection and RA control. Satellite selection determines which effective channels constitute the received multi-satellite channel and therefore shapes its strength--separability balance. RA control, in turn, modifies the element-wise directional gains and thereby reshapes both the strengths and relative spatial orientations of the effective channels.

\begin{remark}
\label{rem:gain_separability_tradeoff}
For multi-satellite reception, clustered arrival directions facilitate simultaneous directional-gain alignment across the selected links but tend to produce highly aligned received channel vectors, thereby reducing spatial separability. Conversely, more widely separated directions improve spatial separability but make simultaneous high-gain alignment across all selected links more difficult. This tradeoff can be partially alleviated through RA steering, which reshapes the element-wise directional gains and allows spatially complementary serving sets to retain sufficiently strong effective channels.
\end{remark}

\subsection{Optimization for Satellite Selection}
\label{subsec:submodular_selection}

Motivated by the preceding geometric insights, we develop a capacity-aware satellite-selection method. With the RA trajectory fixed, the effective-throughput objective decomposes across epochs as
\begin{equation}
    \mathcal{T}(\{\mathcal{S}_\ell\})=\sum_{\ell=1}^{L}F_\ell(\mathcal{S}_\ell),
\end{equation}
where, for epoch $\ell$ and any $\mathcal{S}\subseteq\mathcal{A}_\ell$, the corresponding epoch-level throughput is defined as
\begin{equation}
    F_\ell(\mathcal{S})\triangleq \frac{B\Delta t}{T_{\rm ob}\ln2}\sum_{n\in\mathcal{N}_\ell}\big(\ln\det\bm{R}_{\mathcal{S}}[n]-\ln\det\bm{R}_0[n]\big).
    \label{eq:epoch_set_function}
\end{equation}
Since both the objective and the satellite-selection constraints are separable across epochs, all epoch-level selection problems can be solved independently and in parallel. Hence, for each epoch $\ell$, the satellite-selection problem reduces to
\begin{subequations}
\label{prob:selection_design}
\begin{align}
    \max_{\mathcal{S}}\quad &F_\ell(\mathcal{S}) \label{prob:selection_design_obj}\\
    \mathrm{s.t.}\quad &\mathcal{S}\subseteq\mathcal{A}_\ell,\quad |\mathcal{S}|\leq K_{\max}.
    \label{prob:selection_constraints_v1}
\end{align}
\end{subequations}
Problem~\eqref{prob:selection_design} is a cardinality-constrained set optimization problem, which is generally NP-hard. To solve this problem, we first derive the conditional marginal gain of each candidate satellite, and then establish the monotone submodularity of $F_\ell$. Finally, based on the above characterization, we construct a tight modular surrogate function for $F_\ell$ and optimize it within a discrete MM framework.

\subsubsection{Marginal-Gain Characterization and Submodularity} For a candidate satellite $s\notin\mathcal{S}$, define its conditional marginal gain as $\Delta_\ell(s\mid\mathcal{S})\triangleq F_\ell(\mathcal{S}\cup\{s\})-F_\ell(\mathcal{S})$. Adding $s$ contributes the rank-one covariance term $P_s[n]\bm{h}_s[n]\bm{h}_s^H[n]$ to the current received covariance, i.e.,
\begin{equation}
    \bm{R}_{\mathcal{S}\cup\{s\}}[n]=\bm{R}_{\mathcal{S}}[n]+P_s[n]\bm{h}_s[n]\bm{h}_s^H[n].
\end{equation}
Applying the matrix determinant lemma gives
\begin{align}
    &\Delta_\ell(s\mid\mathcal{S}) =\frac{B\Delta t}{T_{\rm ob}\ln2}\sum_{n\in\mathcal{N}_\ell}\ln\!\left(\frac{\det\bm{R}_{\mathcal{S}\cup\{s\}}[n]}{\det\bm{R}_{\mathcal{S}}[n]}\right)\nonumber\\
    &\quad=\frac{B\Delta t}{T_{\rm ob}\ln2}\sum_{n\in\mathcal{N}_\ell}\ln\!\left(1+P_s[n]\bm{h}_s^H[n]\bm{R}_{\mathcal{S}}^{-1}[n]\bm{h}_s[n]\right).
    \label{eq:marginal_capacity}
\end{align}
Let $\bm{g}_s[n]\triangleq\sqrt{P_s[n]}\bm{R}_0^{-1/2}[n]\bm{h}_s[n]$ denote the interference-whitened effective channel of candidate $s$, and let $\bm{G}_{\mathcal{S}}[n]$ collect the whitened effective channels of the satellites in $\mathcal{S}$ as its columns. Then, by using the Woodbury identity, the quadratic term in \eqref{eq:marginal_capacity} can be rewritten as
\begin{align}\label{eq:marginal_strength_redundancy}
    &P_s[n]\bm{h}_s^H[n]\bm{R}_{\mathcal{S}}^{-1}[n]\bm{h}_s[n] =\|\bm{g}_s[n]\|_2^2\notag \\
    &\quad -\bm{g}_s^H[n]\bm{G}_{\mathcal{S}}[n] \big(\bm{I}+\bm{G}_{\mathcal{S}}^H[n]\bm{G}_{\mathcal{S}}[n]\big)^{-1} \bm{G}_{\mathcal{S}}^H[n]\bm{g}_s[n],
\end{align}
where the first term measures the candidate's interference-whitened channel strength, whereas the second quantifies its overlap with the selected-channel subspace and hence the loss of spatial separability. Therefore, a large conditional marginal gain favors candidates with both strong whitened channels and high spatial separability from the currently selected satellites, consistent with Proposition~\ref{prop:two_satellite_capacity}.

\begin{proposition}[Submodularity of satellite selection]
\label{prop:submodularity}
The epoch-level throughput function $F_\ell(\mathcal{S})$ is monotone and submodular on $\mathcal{A}_\ell$.
\end{proposition}

\begin{IEEEproof}
From \eqref{eq:marginal_capacity}, $\Delta_\ell(s\mid\mathcal{S})\geq0$ because $\bm{R}_{\mathcal{S}}[n]\succ\bm{0}$ and $P_s[n]\geq0$. Hence, $F_\ell$ is monotone. To prove submodularity, consider $\mathcal{U}\subseteq\mathcal{V}$ and $s\notin\mathcal{V}$. Then,
\begin{equation}
    \bm{R}_{\mathcal{V}}[n]-\bm{R}_{\mathcal{U}}[n]=\sum_{j\in\mathcal{V}\setminus\mathcal{U}}P_j[n]\bm{h}_j[n]\bm{h}_j^H[n]\succeq\bm{0}.
\end{equation}
Thus, $\bm{R}_{\mathcal{U}}[n]\preceq\bm{R}_{\mathcal{V}}[n]$, and inversion reverses the positive-definite order, giving
\begin{equation}
    \bm{R}_{\mathcal{U}}^{-1}[n]\succeq\bm{R}_{\mathcal{V}}^{-1}[n].
\end{equation}
Substituting this relation into \eqref{eq:marginal_capacity} yields
\begin{equation}
    \Delta_\ell(s\mid\mathcal{U})\geq\Delta_\ell(s\mid\mathcal{V}),
\end{equation}
which establishes the submodularity of $F_\ell$.
\end{IEEEproof}

Proposition~\ref{prop:submodularity} has two key implications for solving the satellite-selection problem \eqref{prob:selection_design}. First, monotonicity implies that the optimal set is $\mathcal{A}_\ell$ when $|\mathcal{A}_\ell|\leq K_{\max}$, and can be chosen with cardinality $K_{\max}$ otherwise. Second, submodularity shows that the marginal contribution of a candidate is set-dependent and diminishes as the serving set grows. Building on the marginal gain and the submodularity of $F_\ell$, we next construct a modular surrogate for discrete MM optimization.

\subsubsection{Modular Surrogate Construction and Discrete MM Optimization}

For epochs with $|\mathcal{A}_\ell|\leq K_{\max}$, Proposition~\ref{prop:submodularity} implies that all candidate satellites can be selected directly. We therefore focus on the nontrivial case $|\mathcal{A}_\ell|>K_{\max}$. To construct a tight modular surrogate at $\mathcal{S}_\ell^{(r)}$, we form a permutation $\pi_\ell^{(r)}$ of $\mathcal{A}_\ell$ whose first $K_{\max}$ elements constitute $\mathcal{S}_\ell^{(r)}$. The corresponding chain prefixes are defined as
\begin{align}
    \mathcal{P}_{\ell,0}^{(r)}&\triangleq\varnothing, \label{eq:mm_chain1}\\
    \mathcal{P}_{\ell,i}^{(r)}&\triangleq\{\pi_\ell^{(r)}(1),\ldots,\pi_\ell^{(r)}(i)\},\quad i=1,\ldots,|\mathcal{A}_\ell|. \label{eq:mm_chain2}
\end{align}
The incumbent elements are ordered greedily according to their conditional marginal gains relative to the current chain prefix. After the first $K_{\max}$ positions have formed $\mathcal{S}_\ell^{(r)}$, the same rule is continued over the remaining candidates to complete the permutation for all available satellites. For each chain element, define the weight
\begin{align}\label{eq:mm_chain_weight}
    \zeta_{\ell,\pi_\ell^{(r)}(i)}^{(r)}&\triangleq F_\ell(\mathcal{P}_{\ell,i}^{(r)})-F_\ell(\mathcal{P}_{\ell,i-1}^{(r)})=\Delta_\ell\big(\pi_\ell^{(r)}(i)\mid\mathcal{P}_{\ell,i-1}^{(r)}\big).
\end{align}
Each weight is an exact conditional marginal gain evaluated along the permutation chain. Once the permutation is fixed, these chain-dependent marginal gains become fixed candidate weights. By submodularity, each such weight lower-bounds the marginal contribution evaluated with respect to any subset of its preceding permutation elements. This motivates the modular surrogate
\begin{equation}
    \underline F_\ell^{(r)}(\mathcal{U}) \triangleq \sum_{s\in\mathcal{U}}\zeta_{\ell,s}^{(r)}, \qquad \mathcal{U}\subseteq\mathcal{A}_\ell.
    \label{eq:mm_minorizer}
\end{equation}
\begin{proposition}[Incumbent-tight modular surrogate]
\label{prop:modular_minorization}
The surrogate $\underline F_\ell^{(r)}$ provides a global lower bound on $F_\ell$, i.e.,
\begin{equation}
    \underline F_\ell^{(r)}(\mathcal{U})\leq F_\ell(\mathcal{U}),\quad \forall \mathcal{U}\subseteq\mathcal{A}_\ell,
    \label{eq:global_modular_lower_bound}
\end{equation}
and is tight at the incumbent set, i.e.,
\begin{equation}
    \underline F_\ell^{(r)}(\mathcal{S}_\ell^{(r)})=F_\ell(\mathcal{S}_\ell^{(r)}).
    \label{eq:mm_minorizer_tightness}
\end{equation}
\end{proposition}

\begin{IEEEproof}
Let $R\triangleq|\mathcal{U}|$ and index the elements of $\mathcal{U}$ according to their positions in $\pi_\ell^{(r)}$ as $\pi_\ell^{(r)}(i_1),\ldots,\pi_\ell^{(r)}(i_R)$, where $i_1<\cdots<i_R$. Define $\mathcal{U}_j\triangleq\{\pi_\ell^{(r)}(i_1),\ldots,\pi_\ell^{(r)}(i_j)\}$ with $\mathcal{U}_0=\varnothing$. For $\pi_\ell^{(r)}(i_j)$, $\mathcal{U}_{j-1}$ contains only the $j-1$ selected elements preceding it in the permutation, whereas $\mathcal{P}_{\ell,i_j-1}^{(r)}$ contains all $i_j-1$ preceding elements. Hence, $\mathcal{U}_{j-1}\subseteq\mathcal{P}_{\ell,i_j-1}^{(r)}$, and submodularity gives
\begin{align}
    \Delta_\ell\big(\pi_\ell^{(r)}(i_j)\mid\mathcal{U}_{j-1}\big) \geq \Delta_\ell\big(\pi_\ell^{(r)}(i_j)\mid\mathcal{P}_{\ell,i_j-1}^{(r)}\big) =\zeta_{\ell,\pi_\ell^{(r)}(i_j)}^{(r)}.
\end{align}
Summing over the elements of $\mathcal{U}$ yields
\begin{align}
    F_\ell(\mathcal{U}) &=\sum_{j=1}^{R}\Delta_\ell\big(\pi_\ell^{(r)}(i_j)\mid\mathcal{U}_{j-1}\big) \geq\sum_{j=1}^{R}\zeta_{\ell,\pi_\ell^{(r)}(i_j)}^{(r)} =\underline F_\ell^{(r)}(\mathcal{U}),
\end{align}
which proves \eqref{eq:global_modular_lower_bound}. Besides, since $\mathcal{S}_\ell^{(r)}=\mathcal{P}_{\ell,K_{\max}}^{(r)}$, the first $K_{\max}$ chain increments telescope to $F_\ell(\mathcal{S}_\ell^{(r)})-F_\ell(\varnothing)=F_\ell(\mathcal{S}_\ell^{(r)})$, proving \eqref{eq:mm_minorizer_tightness}.
\end{IEEEproof}

Proposition~\ref{prop:modular_minorization} converts the set-dependent marginal structure of $F_\ell$ into additive candidate weights while preserving the exact objective value at the incumbent. Therefore, the next serving set can be obtained by maximizing this modular surrogate under the cardinality constraint:
\begin{equation}
    \mathcal{S}_\ell^{(r+1)}\in\arg\max_{\substack{\mathcal{U}\subseteq\mathcal{A}_\ell,|\mathcal{U}|=K_{\max}}}\sum_{s\in\mathcal{U}}\zeta_{\ell,s}^{(r)}.
    \label{eq:discrete_mm_update}
\end{equation}
Since the surrogate is modular, \eqref{eq:discrete_mm_update} is solved exactly by selecting the $K_{\max}$ candidates with the largest chain weights. Building on the above discussion, we develop a discrete MM framework that iteratively updates the serving set by optimizing an incumbent-tight modular surrogate under the cardinality constraint, as summarized in Algorithm~\ref{alg:satellite_selection}.

\begin{algorithm}[t] \small
\caption{Discrete MM Algorithm for Satellite Selection}
\label{alg:satellite_selection}
\begin{algorithmic}[1]
\REQUIRE Initial sets $\{\mathcal{S}_\ell^{(0)}\}$ and tolerance $\epsilon_{\rm z}$.
\FOR{$\ell=1,\ldots,L$ \textbf{in parallel}}
    \STATE Set $r\leftarrow0$.
    \REPEAT
        \STATE Construct the incumbent-anchored greedy permutation $\pi_\ell^{(r)}$ and calculate the chain weights $\{\zeta_{\ell,s}^{(r)}\}$ via \eqref{eq:mm_chain1}--\eqref{eq:mm_chain_weight}.
        \STATE Obtain $\mathcal{S}_\ell^{(r+1)}$ by solving \eqref{eq:discrete_mm_update}.
        \STATE Set $r\leftarrow r+1$.
    \UNTIL{$F_\ell(\mathcal{S}_\ell^{(r)})-F_\ell(\mathcal{S}_\ell^{(r-1)})\leq\epsilon_{\rm z}$}
    \STATE Set $\mathcal{S}_\ell\leftarrow\mathcal{S}_\ell^{(r)}$.
\ENDFOR
\RETURN $\{\mathcal{S}_\ell\}_{\ell=1}^{L}$.
\end{algorithmic}
\end{algorithm}

\begin{proposition}[Convergence of Algorithm~\ref{alg:satellite_selection}]
\label{prop:discrete_mm_monotonicity}
For each epoch $\ell$, Algorithm~\ref{alg:satellite_selection} generates a nondecreasing sequence of epoch throughput values, i.e.,
\begin{equation}
    F_\ell(\mathcal{S}_\ell^{(r+1)})\geq F_\ell(\mathcal{S}_\ell^{(r)}),\quad \forall r.
    \label{eq:discrete_mm_monotonicity}
\end{equation}
Moreover, the resulting objective sequence is guaranteed to converge.
\end{proposition}

\begin{IEEEproof}
By Proposition~\ref{prop:modular_minorization}, the modular surrogate $\underline F_\ell^{(r)}$ globally lower-bounds $F_\ell$ and is tight at $\mathcal{S}_\ell^{(r)}$. Since \eqref{eq:discrete_mm_update} maximizes this surrogate, we have
\begin{align}
    F_\ell(\mathcal{S}_\ell^{(r+1)}) &\geq \underline F_\ell^{(r)}(\mathcal{S}_\ell^{(r+1)}) \geq \underline F_\ell^{(r)}(\mathcal{S}_\ell^{(r)}) =F_\ell(\mathcal{S}_\ell^{(r)}),
\end{align}
which proves \eqref{eq:discrete_mm_monotonicity}. Since the feasible serving-set family is finite, the nondecreasing objective sequence is bounded above and therefore converges.
\end{IEEEproof}

\subsection{Optimization for RA Trajectory}
\label{subsec:ra_update}
The RA update complements satellite selection by reshaping the effective channel strengths and spatial relationships. With the serving sets ${\mathcal{S}_\ell}$ fixed, the remaining optimization variables are the RA boresight trajectories. Accordingly, the RA-trajectory optimization problem is given by
\begin{subequations}
\label{prob:ra_design}
\begin{align}
    \max_{\bm{F}_{\rm RA}}\quad &\mathcal{T}(\bm{F}_{\rm RA}) \label{prob:ra_design_obj_part}\\
    \mathrm{s.t.}\quad &\bm{f}_m[\ell]\in\mathcal{F}_{\rm RA}, \quad \forall m,\ell, \label{prob:ra_orientation_constraints_part}\\
    &\bm{f}_m^T[\ell-1]\bm{f}_m[\ell] \geq\cos\Omega_{\rm g}, \quad \forall m,\ \ell\geq2.
    \label{prob:ra_slew_constraints_part}
\end{align}
\end{subequations}
Although the objective is separable across epochs, the slew constraints \eqref{prob:ra_slew_constraints_part} couple adjacent RA configurations. We therefore derive feasible local RA updates with neighboring epochs fixed and exploit the chain-structured coupling to organize them into a two-color parallel update scheme.

\subsubsection{Epoch-Level RA Gradient} For compactness, write the $m$th channel coefficient of source $r$ as
\begin{equation}
    h_{m,r}[n] =\upsilon_{m,r}[n][x_{m,r}[\ell,n]]_+^p,
    \label{eq:compact_channel_entry}
\end{equation}
where $\upsilon_{m,r}[n]\triangleq\sqrt{\kappa_{\max}}c_r[n][\bm{a}_r[n]]_m$ and $x_{m,r}[\ell,n]\triangleq\bm{f}_m^T[\ell]\bm{d}_r[n]$. Its derivative with respect to $\bm{f}_m[\ell]$ is given by
\begin{equation}
    \bm{b}_{m,r}[\ell,n] \triangleq
    \begin{cases}
    \upsilon_{m,r}[n]p x_{m,r}^{p-1}[\ell,n]\bm{d}_r[n], & x_{m,r}[\ell,n]>0,\\
    \bm{0}, & x_{m,r}[\ell,n]\leq0.
    \end{cases}
    \label{eq:channel_entry_gradient}
\end{equation}
At the front-side boundary, the derivative is set to zero. We further define
\begin{equation}
\begin{split}
    u_{m,s}[n] &\triangleq [\bm{R}_{\mathcal{S}_\ell}^{-1}[n]\bm{h}_s[n]]_m,\\
    v_{m,q}[n] &\triangleq [(\bm{R}_{\mathcal{S}_\ell}^{-1}[n]-\bm{R}_0^{-1}[n]) \bm{h}_q[n]]_m.
\end{split}
    \label{eq:gradient_auxiliary_variables}
\end{equation}
Using $\mathrm{d}\ln\det\bm{X}=\operatorname{tr}(\bm{X}^{-1}\mathrm{d}\bm{X})$, the Euclidean gradient of the effective throughput with respect to $\bm{f}_m[\ell]$ is derived as
\begin{equation}
\begin{split}
    \bm{g}_{m,\ell}^{\rm E} =\frac{2B\Delta t}{T_{\rm ob}\ln2} &\sum_{n\in\mathcal{N}_\ell} \Bigg\{ \sum_{s\in\mathcal{S}_\ell}P_s[n]\, \Re\!\left(u_{m,s}^*[n]\bm{b}_{m,s}[\ell,n]\right)\\
    +& \sum_{q\in\mathcal{I}_\ell}P_q[n]\, \Re\!\left(v_{m,q}^*[n]\bm{b}_{m,q}[\ell,n]\right) \Bigg\}.
\end{split}
    \label{eq:euclidean_ra_gradient}
\end{equation}
The two sums account for the scheduled channels and orientation-dependent external interference, respectively, with the latter vanishing when $\mathcal{I}_\ell=\varnothing$ or the interference covariance is orientation independent. Projecting onto the tangent space of the unit sphere gives
\begin{equation}
    \bm{g}_{m,\ell}^{\rm R} = (\bm{I}_3-\bm{f}_m[\ell]\bm{f}_m^T[\ell]) \bm{g}_{m,\ell}^{\rm E},
    \label{eq:tangent_ra_gradient}
\end{equation}
which is the Riemannian gradient under the inherited Euclidean metric.

\subsubsection{Slew-Feasible Per-Epoch RA Update} With the neighboring boresights fixed, define $\mathcal{L}_\ell\triangleq\{\ell-1,\ell+1\}\cap\{1,\ldots,L\}$ and the local feasible set
\begin{equation}
\begin{split}
    \mathcal{K}_{m,\ell}\triangleq \big\{\bm{x}\in\mathbb{R}^3:\ & \bm{x}\in\mathcal{F}_{\rm RA}, \,\,\bm{f}_m^T[\ell']\bm{x}\geq\cos\Omega_{\rm g}, \,\, \forall \ell'\in\mathcal{L}_\ell \big\}.
\end{split}
    \label{eq:local_cap_intersection}
\end{equation}
The update direction is obtained from the linear oracle
\begin{equation}
    \bm{s}_{m,\ell} \in \arg\max_{\bm{x}\in\mathcal{K}_{m,\ell}} \left(\bm{g}_{m,\ell}^{\rm R}\right)^T\bm{x}.
    \label{eq:local_linear_oracle}
\end{equation}
To solve this oracle, we enumerate the possible sets of linear constraints that may be active at its optimizer. Let $I$ denote one such set of active constraints. Accordingly, $\bm{A}_I$ collects the corresponding constraint normals, while $\bm{b}_I$ collects their boundary values. Define
\begin{equation}
\begin{split}
    \bm{x}_I &=\bm{A}_I^T(\bm{A}_I\bm{A}_I^T)^\dagger\bm{b}_I,\\
    \bm{P}_I &=\bm{I}_3-\bm{A}_I^T(\bm{A}_I\bm{A}_I^T)^\dagger\bm{A}_I.
\end{split}
    \label{eq:active_set_components}
\end{equation}
For each feasible active set $I$ satisfying $\|\bm{x}_I\|_2\leq1$ and $\bm{P}_I\bm{g}_{m,\ell}^{\rm R}\neq\bm{0}$, a candidate solution to \eqref{eq:local_linear_oracle} is
\begin{equation}
    \widehat{\bm{x}}_I = \bm{x}_I+ \sqrt{1-\|\bm{x}_I\|_2^2}\, \frac{\bm{P}_I\bm{g}_{m,\ell}^{\rm R}} {\|\bm{P}_I\bm{g}_{m,\ell}^{\rm R}\|_2}.
    \label{eq:active_set_candidate}
\end{equation}
Among all these feasible candidates, the one achieving the largest linear objective is selected as the oracle solution $\bm{s}_{m,\ell}$. For $\alpha\in[0,\bar\alpha]$ with $\bar\alpha\in(0,1/2)$, the RA boresight is then updated toward $\bm{s}_{m,\ell}$ as
\begin{equation}
    \bm{f}_m^+[\ell;\alpha] = \frac{(1-\alpha)\bm{f}_m[\ell] +\alpha\bm{s}_{m,\ell}} {\|(1-\alpha)\bm{f}_m[\ell] +\alpha\bm{s}_{m,\ell}\|_2}.
    \label{eq:normalized_chord_update}
\end{equation}
For $\theta_{\max}<\pi/2$ and $\Omega_{\rm g}\leq\pi/2$, the numerator is a convex combination of two feasible points and satisfies all nonnegative linear lower bounds. Its normalization therefore preserves the steering-cap and neighboring slew constraints. Moreover,
\begin{equation}
    \left\|(1-\alpha)\bm{f}_m[\ell]+\alpha\bm{s}_{m,\ell}\right\|_2 \geq1-2\alpha>0,
\end{equation}
thus the update is always well defined. Define the epoch oracle gap as
\begin{equation}
    \mathcal{G}_\ell \triangleq \sum_{m=1}^{M} \left(\bm{g}_{m,\ell}^{\rm R}\right)^T (\bm{s}_{m,\ell}-\bm{f}_m[\ell]) \geq0.
    \label{eq:epoch_product_gap}
\end{equation}
Because the current boresight is feasible for the oracle, $\mathcal{G}_\ell$ is nonnegative and equals the directional derivative of the epoch objective at $\alpha=0$. When $\mathcal{G}_\ell>0$, an Armijo search over $\alpha=\bar\alpha\beta^j$, $j=0,1,\ldots$, finds a feasible step satisfying
\begin{equation}
    F_\ell(\mathcal{S}_\ell;\bm{F}_{\rm RA}^+[\ell]) \geq F_\ell(\mathcal{S}_\ell;\bm{F}_{\rm RA}[\ell]) +\sigma_{\rm A}\alpha\mathcal{G}_\ell,
    \label{eq:armijo_condition}
\end{equation}
where $\beta,\sigma_{\rm A}\in(0,1)$.

\subsubsection{Two-Color Parallel RA-Trajectory Updates}
The preceding development provides a feasible per-epoch update with neighboring RA configurations fixed. Since the inter-epoch slew constraints form a chain, the corresponding coupling graph is bipartite, allowing all odd epochs to be updated independently and in parallel with the even epochs fixed, and vice versa. This odd--even procedure reduces the serial depth from $L$ epoch updates to two parallel rounds while preserving feasibility and monotonic ascent. The resulting RA-trajectory update is summarized in Algorithm~\ref{alg:ra_update}.

\begin{algorithm}[t] \small
\caption{Two-Color RA-Trajectory Update}
\label{alg:ra_update}
\begin{algorithmic}[1]
\REQUIRE Serving sets $\{\mathcal{S}_\ell\}$ and a feasible RA trajectory $\bm{F}_{\rm RA}$.
\FOR{$c\in\{\mathrm{odd},\mathrm{even}\}$}
    \FOR{epochs $\ell$ of color $c$ \textbf{in parallel}}
        \STATE Compute the Riemannian gradients $\{\bm{g}_{m,\ell}^{\rm R}\}_{m=1}^{M}$.
        \STATE Solve \eqref{eq:local_linear_oracle} for $\{\bm{s}_{m,\ell}\}_{m=1}^{M}$ and compute $\mathcal{G}_\ell$ from \eqref{eq:epoch_product_gap}.
        \IF{$\mathcal{G}_\ell>0$}
            \STATE Find $\alpha$ by Armijo search satisfying \eqref{eq:armijo_condition}.
            \STATE Update $\bm{f}_m[\ell]\leftarrow\bm{f}_m^+[\ell;\alpha]$, $\forall m$.
        \ENDIF
    \ENDFOR
\ENDFOR
\RETURN $\bm{F}_{\rm RA}$.
\end{algorithmic}
\end{algorithm}

\subsection{Overall Algorithm, Convergence, and Complexity}
\label{subsec:overall_algorithm}

The preceding discrete-MM satellite-selection procedure and slew-feasible RA-trajectory optimization procedure are integrated into an alternating framework. Starting from feasible serving sets and an RA trajectory, the two blocks are applied successively until the effective-throughput improvement falls below a prescribed tolerance, as summarized in Algorithm~\ref{alg:overall}.

\textit{\textbf{Convergence Analysis:}}
Algorithm~\ref{alg:overall} generates a feasible and nondecreasing sequence of effective-throughput values. For satellite selection, the modular surrogate globally lower-bounds the epoch objective and is tight at the current serving set. Maximizing this surrogate therefore guarantees a nondecreasing epoch objective, as established in Proposition~\ref{prop:discrete_mm_monotonicity}. For RA optimization, each update preserves the steering and slew constraints by construction, while the Armijo condition guarantees nondecreasing throughput. Hence, neither optimization block decreases $\mathcal{T}$. Since the feasible RA-trajectory set is compact, the serving-set space is finite, and the effective throughput is bounded above, the objective sequence generated by Algorithm~\ref{alg:overall} converges to a finite limit.

\textit{\textbf{Computational Complexity Analysis:}}
For satellite selection, each conditional marginal-gain evaluation in \eqref{eq:marginal_capacity} can be implemented via rank-one inverse updates with complexity $\mathcal{O}(M^2)$ per slot. Constructing the complete chain permutation thus requires $\mathcal{O}(|\mathcal{N}_\ell||\mathcal{A}_\ell|^2M^2)$ operations per discrete-MM iteration for epoch $\ell$, up to the initial covariance factorization, while the modular update \eqref{eq:discrete_mm_update} incurs negligible additional cost. The selection problems are independent across epochs and can be processed in parallel. For RA optimization, the dominant cost comes from the covariance factorizations and linear solves required for the gradients, with complexity $\mathcal{O}\!\left(|\mathcal{N}_\ell|\left[M^3+(K_{\max}+|\mathcal{I}_\ell|)M^2\right]\right)$ per epoch. Each local oracle \eqref{eq:local_linear_oracle} is three dimensional and has constant complexity per RA element. The two-color structure further allows same-color epochs to be updated in parallel, reducing the serial depth of one RA update to two color rounds. Hence, both optimization blocks have polynomial per-iteration complexity and substantial epoch-level parallelism.

\begin{algorithm}[t] \small
\caption{The Overall Algorithm for Joint Satellite Selection and RA-Trajectory Optimization}
\label{alg:overall}
\begin{algorithmic}[1]
\REQUIRE Candidate sets $\{\mathcal{A}_\ell\}$, feasible initial serving sets $\{\mathcal{S}_\ell\}$, a feasible RA trajectory $\bm{F}_{\rm RA}$, and tolerance $\epsilon_{\rm T}>0$.
\STATE Update $\{\mathcal{S}_\ell\}$ using Algorithm~\ref{alg:satellite_selection} under the current $\bm{F}_{\rm RA}$.
\STATE Set $\mathcal{T}^{\rm old}\leftarrow\mathcal{T}(\{\mathcal{S}_\ell\},\bm{F}_{\rm RA})$.
\REPEAT
    \STATE Update $\bm{F}_{\rm RA}$ using Algorithm~\ref{alg:ra_update} with $\{\mathcal{S}_\ell\}$ fixed.
    \STATE Update $\{\mathcal{S}_\ell\}$ using Algorithm~\ref{alg:satellite_selection} under the updated $\bm{F}_{\rm RA}$, initialized from the current serving sets.
    \STATE Set $\mathcal{T}^{\rm new}\leftarrow\mathcal{T}(\{\mathcal{S}_\ell\},\bm{F}_{\rm RA})$ and $\Delta\mathcal{T}\leftarrow\mathcal{T}^{\rm new}-\mathcal{T}^{\rm old}$.
    \STATE Set $\mathcal{T}^{\rm old}\leftarrow\mathcal{T}^{\rm new}$.
\UNTIL{$\Delta\mathcal{T}\leq\epsilon_{\rm T}$}
\RETURN $\{\mathcal{S}_\ell\}$ and $\bm{F}_{\rm RA}$.
\end{algorithmic}
\end{algorithm}

\section{Simulation Results and Discussion}
\label{sec:numerical_results}

We consider the first-generation Starlink Walker--Delta shell $53^\circ\!:\!1584/72/1$, at an altitude of $550$~km.  The GS is located at $(50^\circ\mathrm{N},120^\circ\mathrm{E})$.  Unless otherwise stated, it employs a $3\times3$ half-wavelength UPA, schedules at most $K_{\max}=6$ satellites, and operates at $18.2$~GHz over $100$~MHz. The remaining physical and control parameters are listed in Table~\ref{tab:simulation_parameters}. For each random orbital realization, the serving constellation's reference RAAN and Walker phase, together with the initial Earth-rotation angle, are drawn independently from $[0,2\pi)$. The external constellation is independently generated with configuration $70^\circ:\!1584/72/7$ at an altitude of $600$~km.  External satellites that remain visible throughout an epoch are ranked by epoch-average large-scale path gain, and the strongest $Q_{\rm I}$ satellites are retained.  Their common leakage effective isotropic radiated power (EIRP) $P_{\rm I}^{\rm leak}$ is calibrated separately in each orbital realization to achieve a reference interference-to-noise ratio (INR), defined as
\begin{equation}
    \overline{\mathrm{INR}}_{\rm ref} =\frac{1}{N}\sum_{n=1}^{N} \frac{\operatorname{tr}\!\left(\bm{R}_{\rm I}^{\rm ref}[n]\right)} {M\sigma^2},
    \label{eq:reference_inr}
\end{equation}
where, for $n\in\mathcal{N}_\ell$, $\bm{R}_{\rm I}^{\rm ref}[n]=P_{\rm I}^{\rm leak}\sum_{q\in\mathcal{I}_\ell}\bm{h}_q^{\rm ref}[n](\bm{h}_q^{\rm ref}[n])^H$ uses $\bm{h}_q^{\rm ref}[n]$ evaluated per square-root watt of leakage EIRP, with unit transmit gain and $\bm{f}_m=\bm{e}_z$ for every receive element. The calibrated $P_{\rm I}^{\rm leak}$ is then held common to all schemes in the same orbit realization.

\begin{table}[t]
\centering
\caption{Default Simulation Parameters}
\label{tab:simulation_parameters}
\footnotesize
\renewcommand{\arraystretch}{1.06}
\setlength{\tabcolsep}{4.2pt}
\begin{tabular}{lc}
\hline
\textbf{Parameter} & \textbf{Value} \\
\hline
Serving Walker--Delta shell & $53^\circ\!:\!1584/72/1$ \\
Serving altitude, $H$ & $550$ km \\
GS location and elevation mask, $\epsilon_{\min}$ & $(50^\circ\mathrm{N},120^\circ\mathrm{E})$, $10^\circ$ \\
GS array and spacing & $3\times3$, $d_x=d_y=\lambda/2$ \\
Maximum number of streams, $K_{\max}$ & $6$ \\
Scheduled-satellite transmit power & $25$ dBW \\
Carrier frequency and bandwidth & $18.2$ GHz, $100$ MHz \\
Serving-satellite transmit gain & $38$ dBi \\
Additional clear-sky link loss & $3$ dB \\
Receiver noise temperature & $500$ K \\
RA exponent and maximum tilt & $p=4$, $\theta_{\max}=60^\circ$ \\
Angular velocity and guard duration & $\omega_{\max}=20^\circ/{\rm s}$, $T_{\rm g}=1$ s \\
Slots, slot duration, and epochs & $N=192$, $\Delta t=0.5$ s, $L=8$ \\
Wall-clock observation duration & $T_{\rm ob}=103$ s \\
External interferers and reference INR & $Q_{\rm I}=4$, $\overline{\mathrm{INR}}_{\rm ref}=10$ dB \\
\hline
\end{tabular}
\end{table}

\subsection{Validation of the Strength--Separability Tradeoff}

We isolate the strength--separability tradeoff identified in Proposition~\ref{prop:two_satellite_capacity} and Remark~\ref{rem:gain_separability_tradeoff} using a symmetric geometry with six equal-power satellites uniformly spaced by $60^\circ$ in azimuth on the $H=550$~km serving shell. Their angular spread is parameterized by the common off-zenith angle $\psi$, defined as the angle between each satellite direction and the local zenith. We sweep $\psi$ from $0^\circ$ to $60^\circ$, with $\psi=0^\circ$ denoting the co-directional limit. At each $\psi$, all six links have the same elevation and slant range. Increasing $\psi$ spreads the satellite directions farther apart, improving channel separability at the cost of a longer propagation distance and reduced directional-gain sharing. Fig.~\ref{fig:controlled_strength_separability} shows the two components of the tradeoff and their throughput consequence. For each $\psi$, with the six-satellite set fixed, the RA boresights are optimized for effective throughput using the RA-trajectory update in Sec.~\ref{subsec:ra_update}. Fig.~\ref{fig:controlled_strength_separability}(a) reports the normalized aggregate channel strength $\sum_{s=1}^{6}\|\bm{h}_s(\psi)\|^2/\sum_{s=1}^{6}\|\bm{h}_s(0)\|^2$ and the effective spatial rank of the column-normalized channel matrix $\widetilde{\bm{H}}(\psi)=[\widetilde{\bm{h}}_1(\psi),\ldots,\widetilde{\bm{h}}_6(\psi)]$, where $\widetilde{\bm{h}}_s(\psi)=\bm{h}_s(\psi)/\|\bm{h}_s(\psi)\|$. Letting $\sigma_i(\psi)$ denote the singular values of $\widetilde{\bm{H}}(\psi)$ and $\pi_i(\psi)=\sigma_i^2(\psi)/\sum_j\sigma_j^2(\psi)$, its effective rank is $r_{\rm eff}(\psi)=\exp[-\sum_i\pi_i(\psi)\ln\pi_i(\psi)]$. Fig.~\ref{fig:controlled_strength_separability}(b) reports the resulting effective throughput, reflecting the combined effect of the decreasing channel strength and increasing spatial rank shown in Fig.~\ref{fig:controlled_strength_separability}~(a).

At $\psi=0^\circ$, the six channels are co-directional. Fig.~\ref{fig:controlled_strength_separability}(a) therefore shows the maximum aggregate channel strength but an effective spatial rank of only one, resulting in the low throughput of $1.297$~Gbps in Fig.~\ref{fig:controlled_strength_separability}(b). As $\psi$ increases to $25^\circ$, the normalized channel strength decreases to $0.339$, whereas the effective rank rises to $4.72$. The resulting throughput nevertheless increases to its maximum of $5.005$~Gbps, showing that the spatial-separability gain outweighs the channel-strength loss in this regime. Beyond $\psi=25^\circ$, the effective rank rapidly approaches its six-dimensional ceiling, reaching $5.92$ at $\psi=40^\circ$, while the aggregate channel strength continues to decrease. Consequently, further separation provides diminishing spatial benefit, and the throughput falls to $3.646$~Gbps at $\psi=60^\circ$. Since the RAs are reoptimized at every sweep point, this unimodal behavior reflects the residual strength--separability tradeoff after directional-gain optimization. The intermediate optimum confirms that neither maximizing channel strength nor maximizing spatial separation is generally throughput optimal, thereby motivating their joint evaluation through the log-det objective. With external interference, the same argument applies in the interference-whitened channel space, where useful-channel strength, spatial separability, and interference exposure are jointly accounted for.

\begin{figure}[t]
  \centering
  \subfloat[Channel strength and spatial separability]{
    \includegraphics[width=0.8\columnwidth]{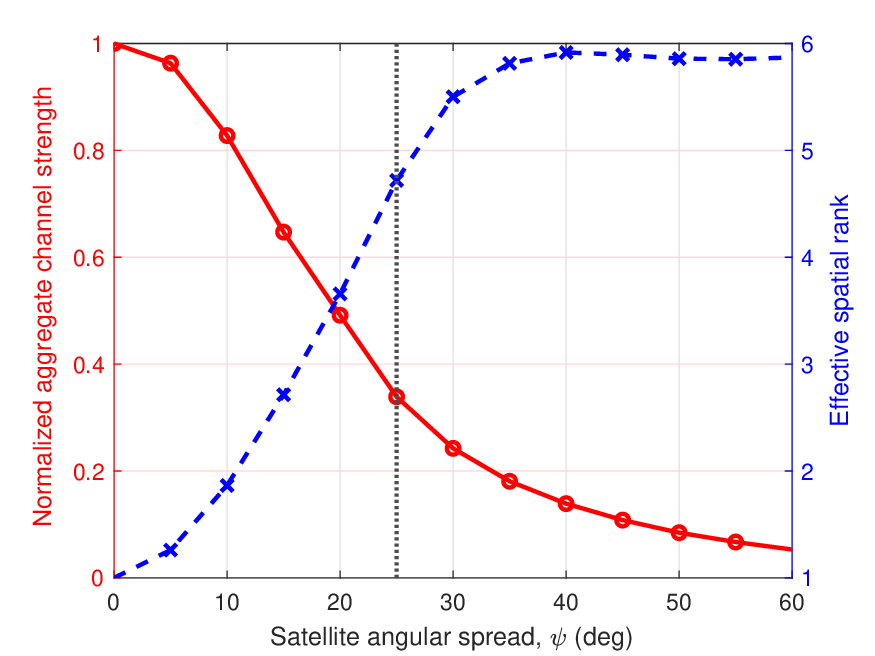}
    \label{fig:controlled_strength_rank}} \\
  \subfloat[Effective throughput]{
    \includegraphics[width=0.8\columnwidth]{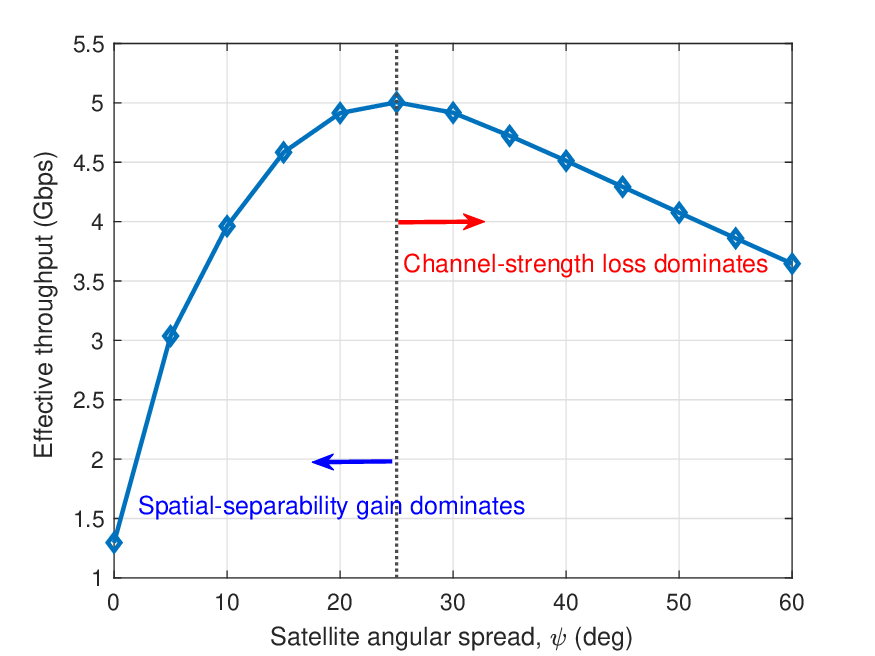}
    \label{fig:controlled_effective_throughput}}
  \caption{Strength--separability tradeoff and the resulting effective throughput with optimized RAs.}
  \label{fig:controlled_strength_separability}
\end{figure}

\subsection{Validation of the Proposed Solution}

Before presenting the performance comparisons, Fig.~\ref{fig:convergence} plots the average effective throughput of the proposed joint algorithm versus the number of iterations for $3\times3$, $4\times4$, and $5\times5$ UPAs. The results are averaged over 200 independent orbital realizations under the default $10$-dB reference INR. For all three array configurations, the average effective throughput increases monotonically and stabilizes after approximately 15 iterations, consistent with the convergence analysis in Section~\ref{subsec:overall_algorithm}. Larger arrays consistently achieve higher throughput, while their similar convergence profiles indicate that the convergence rate remains largely insensitive to the array size. It is also observed that the proposed solution provides more pronounced gains for compact arrays, where joint satellite selection and RA reconfiguration compensate for the more limited intrinsic spatial resolution of the fixed UPA.

\begin{figure}[t]
\centering
\includegraphics[width=0.8\columnwidth]{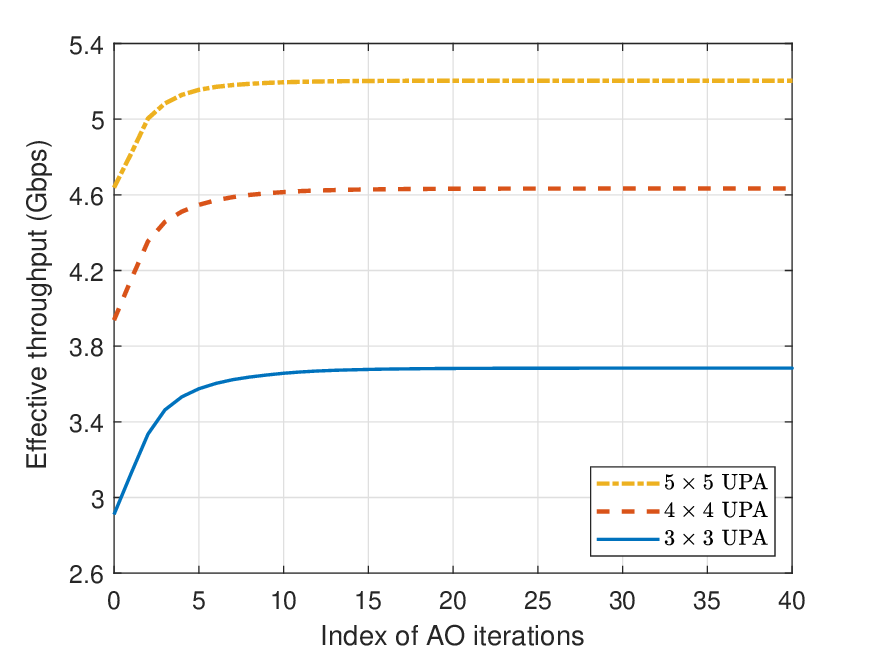}
\caption{Convergence behavior of the proposed Algorithm~\ref{alg:overall}.}
\label{fig:convergence}
\end{figure}

Next, to isolate the individual and joint contributions of satellite selection and RA-trajectory optimization, we compare the following four schemes.

\begin{itemize}
    \item \textbf{RA + MM:} Satellite selection is performed using the discrete MM method in Sec.~\ref{subsec:submodular_selection}, while the RA boresight trajectories are optimized using the two-color update in Sec.~\ref{subsec:ra_update}. The two blocks alternate according to Algorithm~\ref{alg:overall}, constituting the complete proposed scheme.

    \item \textbf{RA + Gain-TopK:} Gain-TopK selects the satellites with the largest epoch-average received powers, while the RA boresight trajectories are optimized using the method proposed in Sec.~\ref{subsec:ra_update}. This benchmark isolates the benefit of capacity-aware satellite selection relative to strength-only selection.

    \item \textbf{Fixed UPA + MM:} All boresights are fixed as $\bm{f}_m[\ell]=\bm{e}_z$, while satellite selection is performed using the discrete MM method in Sec.~\ref{subsec:submodular_selection}. This benchmark isolates the contribution of RA reconfiguration.

    \item \textbf{Fixed UPA + Gain-TopK:} All boresights are fixed as $\bm{f}_m[\ell]=\bm{e}_z$, while Gain-TopK selects up to $K_{\max}$ satellites with the largest epoch-average received powers.
\end{itemize}
All simulation results are averaged over 200 independent random orbital realizations.

\begin{figure}[t]
\centering
\includegraphics[width=0.8\columnwidth]{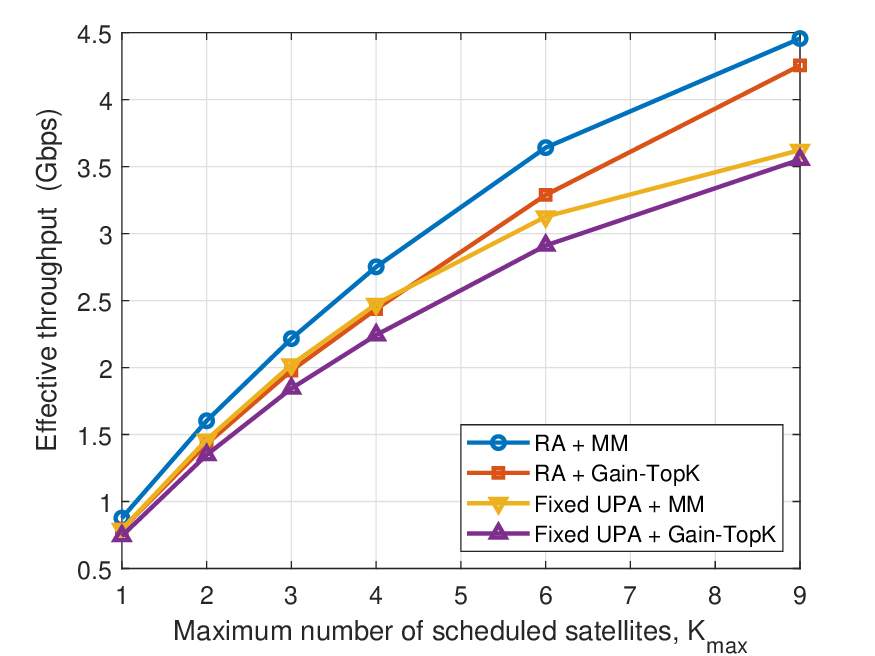}
\caption{Effective throughput versus the number of scheduled satellites.}
\label{fig:kmax}
\end{figure}

Fig.~\ref{fig:kmax} shows the effective throughput versus the maximum number of scheduled satellites. The GS employs the default $3\times3$ half-wavelength UPA, $Q_{\rm I}=4$, and $\overline{\mathrm{INR}}_{\rm ref}=10$~dB. Since the array provides at most $M=9$ spatial dimensions, increasing $K_{\max}$ from one to nine gradually shifts the system from single-stream reception toward full spatial loading. When $K_{\max}=1$, channel separability is not a limiting factor, and the performance differences mainly originate from interference-aware satellite selection and directional-gain optimization. As $K_{\max}$ increases, all schemes improve by exploiting additional spatial dimensions. However, the marginal gain diminishes at large $K_{\max}$ as the available spatial degrees of freedom become saturated. The comparison between \textit{Fixed UPA + MM} and \textit{RA + Gain-TopK} highlights the changing roles of satellite selection and RA reconfiguration. For small-to-moderate $K_{\max}$, \textit{Fixed UPA + MM} is comparable to or slightly better than \textit{RA + Gain-TopK}, indicating that selecting satellites with spatially innovative and interference-compatible channels is particularly important when only part of the available spatial dimensions is utilized. As $K_{\max}$ approaches $M$, however, \textit{RA + Gain-TopK} gradually surpasses \textit{Fixed UPA + MM}. In this regime, satellite selection has progressively less opportunity to introduce new spatial dimensions, whereas RA reconfiguration can still redistribute directional gains and reshape interference coupling across the scheduled links. By jointly exploiting these complementary mechanisms, \textit{RA + MM} consistently achieves the highest throughput. These results indicate that satellite selection makes a relatively larger contribution in the underloaded regime, while RA reconfiguration becomes increasingly important when multiple streams compete for limited spatial resources.

\begin{figure}[t]
\centering
\includegraphics[width=0.8\columnwidth]{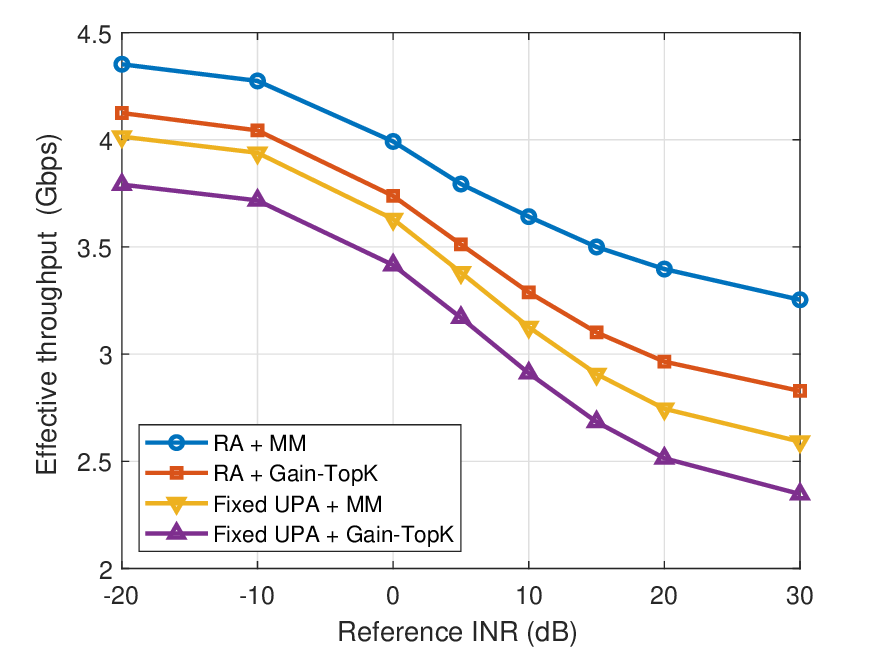}
\caption{Effective throughput versus the reference INR.}
\label{fig:inr}
\end{figure}

Fig.~\ref{fig:inr} compares the schemes as the external-interference level is varied, with $K_{\max}=6$ and $Q_{\rm I}=4$. Although the effective throughput of every scheme decreases as the interference becomes stronger, \textit{RA + MM} degrades more gradually, and its gaps over the two single-component schemes become increasingly pronounced. Under weak interference, interference exposure plays a smaller role in the scheduling decision, making the interference-blind \textit{Gain-TopK} benchmark less disadvantaged. As the INR increases, however, desired-link strength alone becomes increasingly misleading: a satellite that is strong in isolation may occupy a heavily contaminated spatial direction, whereas the proposed \textit{MM} selection directly accounts for the interference-aware log-det objective. Meanwhile, \textit{RA} optimization reshapes the directional gains of both the desired and interfering links. Consequently, MM-based satellite selection provides a larger benefit when combined with optimized RAs than with the fixed UPA, while RA optimization is more effective when paired with MM selection than with \textit{Gain-TopK}. The widening gaps therefore reveal that satellite selection and RA reconfiguration become increasingly complementary as external interference intensifies.

\begin{figure}[t]
\centering
\includegraphics[width=0.8\columnwidth]{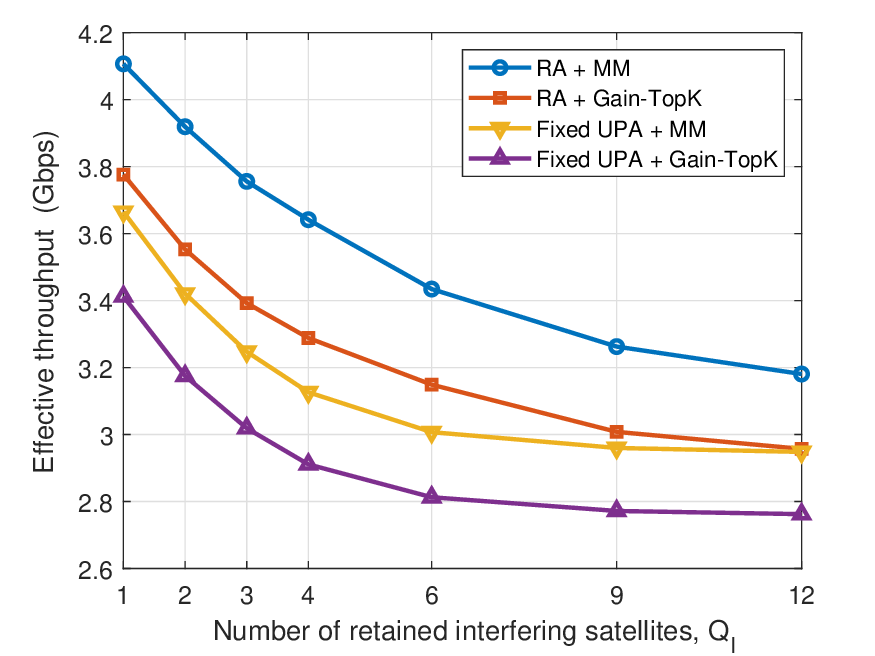}
\caption{Effective throughput versus the number of interfering satellites.}
\label{fig:qi}
\end{figure}

Fig.~\ref{fig:qi} examines the effect of the number of retained interfering satellites while fixing $\overline{\mathrm{INR}}_{\rm ref}=10$~dB. For each $Q_{\rm I}$, the leakage EIRP is recalibrated according to \eqref{eq:reference_inr}, so the figure primarily reflects changes in the spatial structure of the interference rather than its aggregate power. With only a few interferers, the interference covariance is low rank and concentrated in a limited number of arrival directions. MM-based satellite selection can then favor desired channels occupying less contaminated spatial directions, while \textit{RA} reconfiguration can reduce their exposure to the concentrated interference. Consequently, \textit{RA + MM} exhibits its largest advantage in this regime. As $Q_{\rm I}$ increases, the same interference power is distributed over more directions, the interference covariance becomes higher rank, and the throughput of all four schemes decreases. The gaps between the joint scheme and its two single-component counterparts gradually narrow because the increasingly diffuse interference leaves fewer clean directions that can be exploited by either satellite selection or RA steering. Accordingly, \textit{RA + Gain-TopK} and \textit{Fixed UPA + MM} become nearly indistinguishable at large $Q_{\rm I}$, whereas their joint optimization remains consistently superior. The curves eventually flatten as $Q_{\rm I}$ approaches and exceeds the array dimension $M=9$, beyond which additional interferers introduce progressively less new spatial structure.

\begin{figure}[t]
\centering
\includegraphics[width=0.8\columnwidth]{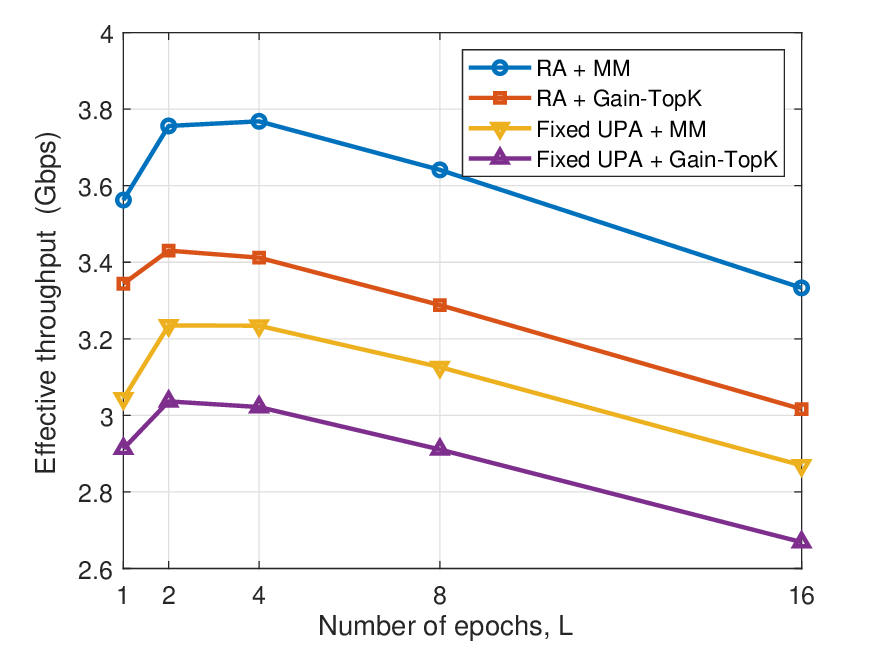}
\caption{Effective throughput versus the number of scheduling epochs.}
\label{fig:num_epoch}
\end{figure}

Fig.~\ref{fig:num_epoch} investigates the impact of the number of scheduling epochs over the fixed observation duration $T_{\rm ob}=103$~s. Increasing $L$ introduces a tradeoff: more frequent schedule and RA updates improve adaptation to the time-varying satellite geometry and interference environment, whereas additional guard intervals reduce the payload duty factor. All schemes initially benefit from dividing the observation window into a small number of epochs. When $L=1$, a single satellite schedule and RA configuration must be maintained throughout the entire interval. The improvement of the two fixed-UPA schemes with increasing $L$ indicates that schedule refreshing alone can already provide effective temporal adaptation. \textit{RA + MM} continues to benefit up to a moderate epoch count because interference-aware satellite selection and RA reconfiguration jointly exploit the evolving spatial channels. However, when $L$ becomes large, the marginal gain from faster adaptation is outweighed by the increasing guard overhead, leading to performance degradation for all schemes. The similar turnover observed for the fixed-UPA schemes confirms that this degradation mainly results from the frame-structure overhead rather than the RA optimization itself. The performance advantage of \textit{RA + MM} is most pronounced at a moderate epoch count, where sufficient temporal flexibility is achieved without excessive reconfiguration cost. These results suggest that, under the considered LEO dynamics and guard duration, a moderate update timescale is sufficient, while overly frequent reconfiguration provides limited additional benefit.

\begin{figure}[t]
\centering
\includegraphics[width=0.8\columnwidth]{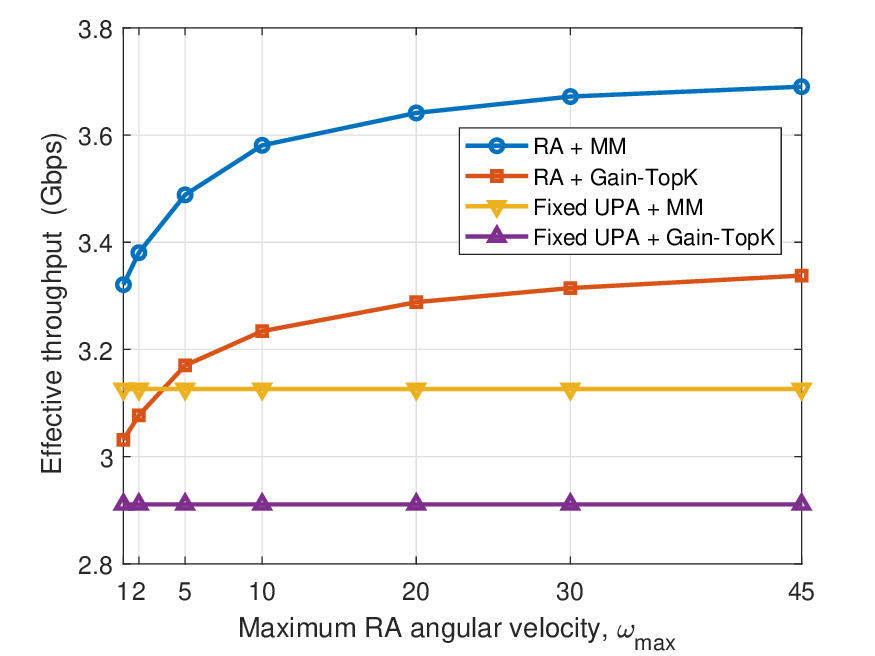}
\caption{Effective throughput versus the maximum RA angular velocity.}
\label{fig:angular_velocity}
\end{figure}

Fig.~\ref{fig:angular_velocity} examines the effect of the maximum RA angular velocity, which limits the boresight displacement allowed between consecutive epochs. As $\omega_{\max}$ increases, both RA-enabled schemes improve monotonically, whereas the two \textit{fixed-UPA} schemes remain unchanged, confirming that the observed gains arise from increased RA actuation flexibility. A particularly revealing crossover occurs at low angular velocities: \textit{Fixed UPA + MM} outperforms \textit{RA + Gain-TopK}, showing that interference-aware satellite selection can be more valuable than slowly reconfigurable RAs applied to a strength-only schedule. More notably, \textit{RA + MM} at the lowest tested angular velocity slightly outperforms \textit{RA + Gain-TopK} even at the highest angular velocity. Thus, faster RA actuation cannot fully compensate for a satellite set with poor spatial separability or interference exposure, while the proposed selection method substantially relaxes the required actuator speed. The RA-enabled curves gradually saturate beyond approximately $20^\circ/{\rm s}$, indicating that the allowed inter-epoch motion is then sufficient to follow the evolving favorable boresight directions. Further increasing the angular velocity therefore provides only marginal throughput improvement. These observations demonstrate that satellite selection and RA actuation are complementary rather than interchangeable, while also supporting the practical choice of a moderate RA angular velocity.

\section{Conclusion}
\label{sec:conclusion}

This paper investigated joint satellite selection and RA boresight-trajectory optimization for concurrent multi-satellite reception in LEO satellite-to-ground communications. We developed a two-timescale framework that accounts for epoch-held serving sets and RA configurations, time-varying satellite geometry, exogenous cochannel interference, and inter-epoch slew constraints. By analyzing the marginal capacity gain in the interference-whitened domain, we revealed the joint roles of effective channel strength and spatial separability and established the monotonicity and submodularity of the satellite-selection objective. Based on this structure, we developed a discrete MM method using an incumbent-tight modular lower-bound surrogate, together with slew-feasible RA updates organized through a two-color parallel update scheme. The resulting alternating algorithm preserves feasibility and guarantees monotonic improvement and convergence of the objective value. Simulations demonstrated consistent gains from the joint design and revealed an optimal balance between channel strength and spatial separability. Satellite selection is particularly important in underloaded and actuator-limited regimes, whereas RA shaping becomes more influential near full spatial loading, with stronger complementarity under severe external interference.

\bibliographystyle{IEEEtran}
\bibliography{bibtex}

\end{document}